\documentclass[nofootinbib,twocolumn,pra,letterpaper,longbibliography]{revtex4-2}
\usepackage{latexsym,amsmath,amssymb,amsfonts,graphicx,color,amsthm}
\usepackage{enumerate}
\usepackage{braket}
\usepackage{adjustbox}
\usepackage{footnote}
\usepackage[dvipsnames]{xcolor}
\usepackage{tipa}
\usepackage{textcomp}
\usepackage{fontenc}
\usepackage{mathrsfs}
\usepackage{color}
\usepackage{colortbl}
\usepackage{graphics,graphicx}
\usepackage{txfonts}
\usepackage{lipsum, babel}
\usepackage{bbm}

\usepackage[unicode=true,pdfusetitle, bookmarks=true,bookmarksnumbered=false,bookmarksopen=false, breaklinks=false,pdfborder={0 0 0},backref=false,colorlinks=false] {hyperref}
\hypersetup{ colorlinks,linkcolor=myurlcolor,citecolor=myurlcolor,urlcolor=myurlcolor}

\definecolor{myurlcolor}{rgb}{0,0,0.7}

\usepackage{float}
\usepackage{hyperref}

\newcommand{\cA}{\mathcal{A}}
\newcommand{\cB}{\mathcal{B}}
\newcommand{\cC}{\mathcal{C}}
\newcommand{\cD}{\mathcal{D}}
\newcommand{\cE}{\mathcal{E}}
\newcommand{\cF}{\mathcal{F}}
\newcommand{\cG}{\mathcal{G}}
\newcommand{\cH}{\mathcal{H}}
\newcommand{\cI}{\mathcal{I}}
\newcommand{\cJ}{\mathcal{J}}
\newcommand{\cK}{\mathcal{K}}
\newcommand{\cL}{\mathcal{L}}
\newcommand{\cM}{\mathcal{M}}
\newcommand{\cN}{\mathcal{N}}
\newcommand{\cO}{\mathcal{O}}
\newcommand{\cP}{\mathcal{P}}
\newcommand{\cQ}{\mathcal{Q}}
\newcommand{\cR}{\mathcal{R}}
\newcommand{\cS}{\mathcal{S}}
\newcommand{\cT}{\mathcal{T}}

\newcommand{\cV}{\mathcal{V}}
\newcommand{\cW}{\mathcal{W}}

\newcommand{\cY}{\mathcal{Y}}
\newcommand{\cZ}{\mathcal{Z}}

\newcommand{\Id}{\mathbbm{1}}
\newcommand{\tr}{Tr}

\newtheorem{theorem}{Theorem}
\newtheorem{proposition}{Proposition}
\newtheorem{lemma}{Lemma}
\newtheorem{corollary}{Corollary}
\newtheorem{definition}{Definition}

\newtheorem{remark}{Remark}

\begin{document}

\title{Distance-based measures and Epsilon-measures for measurement-based quantum resources}
\author{Arindam Mitra$^1$}
\email{amitra36013@kriss.re.kr}
\email{arindammitra143@gmail.com}


\author{Sumit Mukherjee$^1$}

\email{mukherjeesumit93@gmail.com}


\author{Changhyoup Lee$^{2,3,4}$}

\email{changhyoup.lee@gmail.com}

\affiliation{$^1$Korea Research Institute of Standards and Science, Daejeon 34113, South Korea.}
\affiliation{$^2$Department of Physics, Hanyang University, Seoul 04763, Korea,}
\affiliation{$^3$Hanyang Institute for Quantum Science and Quantum Technology, Hanyang University, Seoul 04763, Korea}
\affiliation{$^4$Research Institute of Natural Sciences, Hanyang University, Seoul 04763, Korea}

\date{\today}

\begin{abstract}
Quantum resource theories provide a structured and elegant framework for quantifying quantum resources. While state-based resource theories have been extensively studied, the measurement-based resource theories remain relatively underexplored. In practical scenarios where a quantum state or a set of measurements is only partially known, conventional resource measures often fall short in capturing the resource content. In such cases, $\epsilon$-measures offer a robust alternative, making them particularly valuable. In this work, we investigate the quantification of measurement-based resources using distance-based measures, followed by a detailed analysis of the mathematical properties of $\epsilon$-measures. We also extend our analysis by exploring the connections between  $\epsilon$-measures and some key quantities relevant to resource manipulation tasks. Importantly, the analysis of resources based on sets of measurements are tedious compared to that of single measurements as the former allows more general transformations such as controlled implementation. Yet our framework applies not only to resources associated with individual measurements but also to those arising from sets of measurements.  In short, our analysis is applicable to existing resource theories of measurements and has the potential to be useful for all resource theories of measurements that are yet to be developed.

\end{abstract}

\maketitle

\section{Introduction}
There exist numerous quantum resources that have no classical counterpart \cite{Chitambar_QRT_review}. Examples include entanglement \cite{Horodecki_review_entang}, coherence \cite{Baumgratz_coh_RT,Winter_coh_RT,Bischof_coh_RT}, nonlocality \cite{deVicente_2014,Wolfe2020quantum}, steering \cite{Wiseman_steer,Gallego_steer_RT}, contextuality \cite{abrahmsky2017,amaral19}, measurement incompatibility \cite{Buscemi_meas_incomp}, measurement sharpness \cite{Mitra_meas_sharp,Buscemi_meas_sharp}, and measurement coherence \cite{baek_meas_coh}. Some of these resources are intrinsic properties of quantum states (e.g., entanglement, coherence etc.), while others are associated with quantum measurements (e.g., measurement incompatibility, measurement sharpness, measurement coherence etc). Among the measurement-based resources, some are the properties of individual measurements (e.g., measurement coherence and measurement sharpness), whereas others such as measurement incompatibility are defined for sets of measurements. These resources enable quantum advantages in a variety of information-theoretic tasks \cite{Skrzypczyk_incomp_state_disc,Uola_q_res_exclus,Carmeli_incomp_meas_qrac}. Consequently, quantifying them is of fundamental importance. 

Quantum resource theories offer a principled and elegant framework for quantifying various quantum resources \cite{Chitambar_QRT_review}. While the resource theories of most of the state-based quantum resources have been extensively explored in the literature, only a limited number of measurement-based resource theories have been studied in comparable detail \cite{baek_meas_coh,Buscemi_meas_incomp,Buscemi_meas_sharp,Guff_RT_meas}. To the best of our knowledge, only the resource theories of measurement incompatibility, measurement coherence, and measurement sharpness have received significant attention. Although the framework we develop in this work is general, we focus our analysis primarily on those three measurement-based resource theories.

In practice, usually, neither quantum states nor quantum measurements are perfectly known. Both the prepared quantum states as well as the implementable measurement devices often encounter some noise. Therefore, in the scenarios where the states and the measurements are partially known, it is difficult to perfectly estimate the resource content in those. In such cases, $\epsilon$-measures can be used as a resource quantifier. The $\epsilon$-measures for quantum state-based resources have been studied in the Refs. \cite{Mora_epsilon_meas_entang,Xi_epsilon_meas_coh,Luo_epsilon_meas_state}. However, to the best of our knowledge, while specific distance-based resource measure for measurement based resources have been studied in \cite{Tendick_dist_res_meas}, the $\epsilon$-measures for measurement-based resources have not been studied in detail. In this work, our \emph{motivation} is to study the properties of the generic distance based-measure and $\epsilon$-measures for quantum measurement-based resources.

In this work, we provide a systematic methodology for quantifying measurement-based quantum resources through  generic distance-based measures satisfying a few properties. This framework provides a general approach to evaluate how resouceful a given measurement or a set of measurements is according to a chosen distance-based resource measure. Following this, we study the concept of 
$\epsilon-$measure tailored for measurement-based resource theories and undertake a detailed analysis of their mathematical properties. In addition, we investigate the relationship between these $\epsilon-$measures and several key operational quantities in resource theory, including the resource dilution cost and the smooth regularization of resource measures.

It is important to note that the analysis developed for a generic resource theory of individual measurements cannot be directly extended to a generic resource theory of sets of measurements, as the latter permits more general transformations—such as controlled implementations. Nevertheless, our analysis is applicable not only to the quantification of resources contained in individual measurements, such as measurement coherence or measurement sharpness, but also to the quantification of resources  that emerge from sets of measurements, such as measurement incompatibility.

The rest of the paper is organized as follows. In Sec. \ref{Sec:Prelim}, we discuss the preliminaries. More specifically, in Sec. \ref{Subsec:prelim_QM_QC}, we discuss  the various well-known properties of measurements, channels, superchannels and diamond distance. In Sec. \ref{Subsec:prelim_resource}, we discuss the basic structure of an arbitrary resource theory, some usual assumptions for a resource theory of generic measurement-based resources and define $\epsilon$-measures for a generic measurement-based resource. In Sec. \ref{Sec:Main_results}, we discuss our main results. More specifically, in Sec. \ref{Subsec:Gen_t_set_meas}, we study the transformation of sets of channels, sets of  measurements, define a distance measure for them and study its mathematical properties. In Sec. \ref{Subsec:dist_based_res_meas}, we construct and study a distance-based resource measure for measurement-based resources using a generic distance measure satisfying some properties. In Sec. \ref{Subsec:properties_epsilon_measure}, we study some properties of $\epsilon$-measure for measurement-based resources. In Sec. \ref{Subsec:one_shot_distill_assymp_res_meas}, we study the one-shot dilution cost, smooth asymptotic resource measures and its connection to $\epsilon$-measures for measurement-based resources. In Sec. \ref{Sec:Conc}, we summarize our results and discuss future directions.

\section{Preliminaries}
\label{Sec:Prelim}
\subsection{Quantum Measurements and quantum channels}
\label{Subsec:prelim_QM_QC}
A quantum measurement $M=\{M(x)\in\cL(\cH)\}_{x\in\Omega_M}$ acting on Hilbert space $\cH$  is defined as a set of positive semidefinite matrices acting on Hilbert space $\cH$ such that $\sum_{x\in\Omega_M}M(x)=\Id_{\cH}$ where $\cL(\cH)$ is the set of all linear operators on Hilbert space $\cH$, $\Omega_M$ is the outcome set of $M$ and $\Id_{\cH}$ is the idenitity matrix on the Hilbert space $\cH$ \cite{Heinosaari_book_QF}. In this work, we restrict ourselves to finite dimensional Hilbert spaces and measurements with finite number of outcomes. Here we call $M(x)$ as the POVM elements of $M$ for all $x$. The measurement $M$ is said to be projective if $M^2(x)=M(x)\forall x\in \Omega_M$.
When the measurement $M$ is performed on a system with quantum state $\rho\in\cS(\cH)$, the probability of an arbitrary outcome $x$ is $\tr[\rho M(x)]$ where $\cS(\cH)$ is the set of all density matrices on Hilbert space $\cH$. We denote the set of all measurements acting on Hilbert space $\cH$ as $\mathscr{M}(\cH)$. Given a pair of quantum measurements $M\in\mathscr{M}(\cH)$ and $N\in\mathscr{M}(\cK)$, the measurement $M\otimes N:=\{(M\otimes N)(x,y)=M(x)\otimes N(y)\}_{x\in\Omega_{M},y\in\Omega_{N}}\in\mathscr{M}(\cH\otimes\cK)$. Clearly, the outcome set of the measurement $M\otimes N$ is $\Omega_{M\otimes N}=\Omega_M\times\Omega_N$. We denote the one-outcome trivial measurement acting on Hilbert space $\cH$ as $\cT_{\cH}:=\{\Id_{\cH}\}$. Implementation of $\cT_{\cH}$ is equivalent to performing ``no measurement". Clearly, $\cT_{\cH_1\otimes\cH_2}=\cT_{\cH_1}\otimes\cT_{\cH_2}$
If a quantum measurement $M\in \mathscr{M}(\cH)$ is the probabilistic mixing of a pair of quantum measurements $M_1\in \mathscr{M}(\cH)$ and $M_2\in \mathscr{M}(\cH)$ then $M=\{M(x)=p M_1(x)+(1-p) M_2(x)\}_{x\in\Omega_{M}}$ where $\Omega_{M}=\Omega_{M_1}=\Omega_{M_2}$ and $p$ being the probability. We denote it as $M=p M_1+(1-p) M_2$. We denote a process where a given measurement $M$ (as an input) is probabilistically mixed with a fixed channel $N$ as
\begin{align}
    \cP_{p,N}[M]:=pM+(1-p)N,
\end{align}
where $0\leq p\leq 1$ is the probability.

A quantum channel $\Lambda:\cL(\cH)\rightarrow \cL(\cK)$ is a completely positive trace preserving (CPTP) linear map that transforms a quantum state to another quantum state \cite{Heinosaari_book_QF}. The action of $\Lambda$ in Heisenberg picture is denoted by $\Lambda^{\dagger}:\cL(\cK)\rightarrow\cL(\cH)$ which is defined via the equation

\begin{align}
    \tr[\Lambda(W)Z]=\tr[W\Lambda^{\dagger}(Z)]~\forall W\in\cL(\cH),Z\in\cL(\cK).
\end{align}

We denote the set of all quantum channels with $\cL(\cH)$ as input space and $\cL(\cK)$ as output space as $\mathscr{C}(\cH,\cK)$. If a quantum channel $\Psi:\cL(\cH)\rightarrow\cL(K)$ is composition of two quantum channels $\Lambda_1:\cL(\cH)\rightarrow\cL(\cH_1)$ and $\Lambda_2:\cL(\cH_1)\rightarrow\cL(\cK)$ then for all $\rho\in\cL(\cH)$, $\Psi(\rho)=\Lambda_2(\Lambda_1(\rho))$. We use the symbol $\circ$ and the short-hand notation $\Psi=\Lambda_2\circ\Lambda_1$ to denote this. The shorthand notation such as this will be repeated throughout the paper.

Superchannels are the physically realizable transformations between quantum channels. A super-channel $\hat{\Xi}$ transforming a quantum channel in $\Lambda\in\mathscr{C}(\cH_1,\cH_2)$ into a quantum channel $\hat{\Xi}[\Lambda]\in\mathscr{C}(\cK_1,\cK_2)$ can be written as

\begin{align}
    \hat{\Xi}[\Lambda]=\Theta_{post}\circ(\Lambda\otimes\mathbbm{I}_{\cR})\circ\Theta_{pre},\label{Eq:sup_chan}
\end{align}
where the quantum channel $\Theta_{pre}:\cL(\cK_1)\rightarrow\cL(\cH_{1}\otimes\cR)$ is known as pre-processing and the quantum channel , $\Theta_{post}:\cL(\cH_2\otimes\cR)\rightarrow\cL(\cK_2)$ is known as post-processing and $\mathbbm{I}_{\cR}$ is the identity quantum channel on $\cL(\cR)$ \cite{Gour_compar_sup_chan}.

\begin{remark}
   \rm{ The transformations between two sets of quantum channels are more general. Such transformations allow controlled implementation of the quantum channels in the set, in general. We will discuss this in \ref{Subsec:Gen_t_set_meas}.}
\end{remark}

If a quantum channel $\Theta\in \mathscr{C}(\cH_1,\cH_2)$ is the probabilistic mixing The probabilistic mixing of two quantum channels $\Lambda_1\in \mathscr{C}(\cH_1,\cH_2)$ and $\Lambda_2\in \mathscr{C}(\cH_1,\cH_2)$ then $\Theta(\rho)=p\Lambda_1(\rho)+(1-p)\Lambda_2(\rho)$ for all $\rho\in\cL(\cH)$. We denote it as $\Theta=p\Lambda_1+(1-p)\Lambda_2$. We denote a process where a given channel $\Lambda\in \mathscr{C}(\cH,\cK)$ (as an input) is probabilistically mixed with an arbitrary fixed channel $\Psi\in \mathscr{C}(\cH,\cK)$ as
\begin{align}
    \cP_{p,\Psi}[\Lambda]:=p\Lambda+(1-p)\Psi,
\end{align}
where $0\leq p\leq 1$ is the probability.

Now, given a pair of channels $\Lambda_1:\cL(\cH)\rightarrow\cL(\cK)$ and $\Lambda_2:\cL(\cH)\rightarrow\cL(\cK)$, one can define a distance between them.  One of such a distance measure between $\Lambda_1$ and $\Lambda_1$ can be defined as

\begin{align}
    \cD_{\Diamond}(\Lambda_1, \Lambda_2) &:=\mid\mid \Lambda_1-\Lambda_2\mid\mid_{\Diamond} \nonumber \\
    & =\max_{\rho_{AB}\in\cS(\cH_A\otimes\cH_B)}\mid\mid \Lambda_1\otimes\mathbbm{I}_{\cH_B}(\rho_{AB}) \nonumber \\ & \hspace{3.2cm} -\Lambda_2\otimes\mathbbm{I}_{\cH_B}(\rho_{AB})\mid\mid_1,\label{Eq:Dia_formula_H_p}
\end{align}
 where we denote the trace norm by $\mid\mid.\mid\mid_1$, and $\dim(\cH_A)=\dim(\cH_B)$. This distance is known as the diamond distance.

 \begin{remark}
      \rm{It should be noted that, usually a factor $1/2$ is added in the diamond distance to normalize it \cite{Tendick_dist_res_meas}. But we do not use this factor as it is trivial and does not affect any of our results and therefore, is unnecessary. If $1/2$ factor is included, the upper bound of the diamond distance is $1$ \cite{Watrous_book_TQI}. But as we did not include the $1/2$ factor in the diamond distance, the upper bound is $2$ in our case. It is known that $\cD_{\Diamond}$ satisfies all three conditions (i.e., positivity, symmetric, and triangular inequality) that are required to be satisfied by any distance measure \cite{Watrous_book_TQI}.}\label{Re:diam_dist_half_fact}
 \end{remark}

 It is well-known that $\cD_{\Diamond}$ is monotonically non-increasing under arbitrary pre-processing and post-processing channels, or more generally under an arbitrary super-channel. In other words, for an arbitrary super channel $\hat{\Xi}$ that transforms arbitrary  $\Lambda_i\in\mathscr{C}(\cH_1,\cH_2)$ to $\hat{\Xi}[\Lambda_i]\in\mathscr{C}(\cK_1,\cK_2)$ for $i=\{1,2\}$, we have \cite{Watrous_book_TQI}

 \begin{align}
     \cD_{\Diamond}(\hat{\Xi}[\Lambda_1],\hat{\Xi}[\Lambda_2])\leq\cD_{\Diamond}(\Lambda_1,\Lambda_2)\label{Eq:mono_dec_dimon_sup_chan}.
 \end{align}
 
 Furthermore, for four quantum channels $\Lambda_1,\Lambda_2,\Theta_1,\Theta_2\in\mathscr{C}(\cH,\cK)$ and  $0 \leq p\leq 1$  one can show that \cite{Watrous_book_TQI}

 \begin{align}
     &\cD_{\Diamond}(p\Lambda_1+(1-p)\Lambda_2,p\Theta_1+(1-p)\Theta_2)\nonumber\\
     \leq & p\cD_{\Diamond}(\Lambda_1,\Theta_1)+ (1-p)\cD_{\Diamond}(\Lambda_2,\Theta_2).\label{Eq:dim_dist_conv}
 \end{align}
 
 In other words, diamond distance satisfies joint convexity property. Clearly, $\cD_{\Diamond}$ is also monotonically non-increasing under $\cP_{p,\Psi}$.

\begin{remark}
    \rm{Note that the trace norm of an arbitrary linear map $\Phi:\cL(\cH)\rightarrow\cL(\cK)$ is defined as $\mid\mid\Phi\mid\mid_1:=\rm{max}_{X\in\cL(\cH)}\{\mid\mid \Phi(X)\mid\mid_1\mid~\mid\mid X\mid\mid_1\leq 1\}$. Now, the diamond norm of an arbitrary linear map $\Phi:\cL(\cH)\rightarrow\cL(\cK)$ is originally defined as
    $\mid\mid\Phi\mid\mid_{\Diamond}:=\mid\mid \Phi\otimes\mathbbm{I}_{\cH_B}\mid\mid_1$ with $\dim(\cH)=\dim(\cH_B)$ \cite{Watrous_book_TQI}. But note that $(\Gamma_{M_1}-\Gamma_{M_2})$ is a  Hermiticity-preserving map and we know that for the Hermiticity-preserving map, the original definition coincides with Eq. \eqref{Eq:Dia_formula_H_p} \cite{Watrous_book_TQI}. Furthermore, one can show that $\mid\mid\Phi\mid\mid_{\Diamond}=\mid\mid \Phi\otimes\mathbbm{I}_{\cH_R}\mid\mid_1$ for any $\cH_R$ with $\dim(\cH)\leq\dim(\cH_R)$ \cite{Watrous_book_TQI}. It should be mentioned here that for a pair of arbitrary linear maps $\Phi_1:\cL(\cH_1)\rightarrow \cL(\cK_1)$ and $\Phi_2:\cL(\cH_2)\rightarrow \cL(\cK_2)$, the equality $\mid\mid\Phi_1\otimes\Phi_2\mid\mid_{\Diamond}=\mid\mid\Phi_1\mid\mid_{\Diamond}.\mid\mid\Phi_2\mid\mid_{\Diamond}$ holds \cite{Watrous_book_TQI}. Furthermore, if $\Phi:\cL(\cH)\rightarrow\cL(\cK)$ is a quantum channel, $\mid\mid\Phi\mid\mid_{\Diamond}=1$ \cite{Watrous_book_TQI}. }\label{Re:Diam_norm}
\end{remark} 

Consider the swap unitary channel $\mathbf{SWAP}_{\cH_A\leftrightarrow \cH_B}\in\mathscr{C}(\cH_A\otimes\cH_B,\cH_B\otimes\cH_A)$ such that for all $\rho_{AB}\in\cL(\cH_A\otimes\cH_B)$

\begin{align}
    \mathbf{SWAP}_{\cH_A\leftrightarrow \cH_B}(\rho_{AB})=U_{A\leftrightarrow B}\rho_{AB}U^{\dagger}_{A\leftrightarrow B},
\end{align}
where $U_{A\leftrightarrow B}=\sum_{ij}\ket{j}_{B}\bra{i}_{A}\otimes\ket{i}_{A}\bra{j}_{B}$. Note that for arbitrary $\Lambda_i\in\mathscr{C}(\cH_i,\cK_i)$ for $i\in\{A,B\}$

\begin{align}
    \Lambda_B\otimes\Lambda_A=\mathbf{SWAP}_{\cK_A\leftrightarrow \cK_B}\circ(\Lambda_A\otimes\Lambda_B)\circ\mathbf{SWAP}_{\cH_B\leftrightarrow \cH_A}.
\end{align}

Note that the channel $\mathbbm{I}_{\cR}\otimes\mathbf{SWAP}_{(.)\leftrightarrow (.)}$ is invertible for an arbitrary Hilbert space $\cR$. Then for all $\Lambda_C,\Theta_C\in\mathscr{C}(\cH_C,\cK_C)$, $\Lambda_A,\Theta_A\in\mathscr{C}(\cH_A,\cK_A)$, and $\Lambda_B,\Theta_B\in\mathscr{C}(\cH_B,\cK_B)$ from Eq. \eqref{Eq:mono_dec_dimon_sup_chan} we have

\begin{align}
    &\cD_{\Diamond}(\Lambda_C\otimes \Lambda_A\otimes \Lambda_B,\Theta_C\otimes \Theta_A\otimes \Theta_B)\nonumber\\
    =&\cD_{\Diamond}(\Lambda_C\otimes \Lambda_B\otimes \Lambda_A,\Theta_C\otimes \Theta_B\otimes \Theta_A)
\end{align}
or more generally, for arbitrary $\Lambda_{CAB},\Theta_{CAB}\in\mathscr{C}(\cH_C\otimes\cH_A\otimes\cH_B,\cK_C\otimes\cK_A\otimes\cK_B)$

\begin{align}
    \cD_{\Diamond}(\Lambda_{CAB},\Theta_{CAB})=\cD_{\Diamond}(\Lambda_{CBA},\Theta_{CBA}),\label{Eq:swap_inv_dimon}
\end{align}
where $\Lambda_{CBA}=(\mathbbm{I}_{\cK_C}\otimes\mathbf{SWAP}_{\cK_A\leftrightarrow \cK_B})\circ\Lambda_{CAB}\circ(\mathbbm{I}_{\cH_C}\otimes\mathbf{SWAP}_{\cH_B\leftrightarrow \cH_A})$ and $\Theta_{CBA}=(\mathbbm{I}_{\cK_C}\otimes\mathbf{SWAP}_{\cK_A\leftrightarrow \cK_B})\circ\Theta_{CAB}\circ(\mathbbm{I}_{\cH_C}\otimes\mathbf{SWAP}_{\cH_B\leftrightarrow \cH_A})$. In other words, $\cD_{\Diamond}$ is invariant under swapping of Hilbert spaces.

In Heisenberg picture, a quantum channel $\Lambda^{\dagger}$ transforms a quantum measurement $M=\{M(x)\}$ to another quantum measurement $\Lambda^{\dagger}[M]=\{\Lambda^{\dagger}[M(x)]\}_{x\in\Omega_M}$. The action of a CP trace non-increasing linear map in Heisenberg picture is defined in a similar way.

After performing a measurement $M=\{M(x)\}$, one can post-process the outcomes which is equivalent to performing another measurement $N=N(y)$. For this purpose, if one uses the probability distributions $\nu_x=\{\nu_x(y)\}$ to post-process the outcome $x$ then

\begin{align}
    N(y)=\sum_{x\in\Omega_M}\nu_x(y)M(x)~\forall y\in \Omega_N.
\end{align}

A measurement $M$ can be performed using many different quantum instruments. A quantum instrument $\cI$ is defined as a set of CP trace non-increasing linear maps $\{\Phi_x:\cL(\cH)\rightarrow \cL(\cK)\}_{x\in\hat{\Omega}_{\cI}}$ that sums up to a quantum channel $\Phi=\sum_{x}\Phi_{x}$ where $\hat{\Omega}_{\cI}$ is the outcome set of the instrument $\cI$ \cite{Heinosaari_book_QF}. It is well-known that both quantum measurements and quantum channels are special cases of quantum instruments. If the instrument $\cI$ implements the measurement $M\in\mathscr{M}(\cH)$, then

\begin{align}
    \Phi_x^{\dagger}[\Id_{\cK}]=M(x)
\end{align}
hold for all $x\in \hat{\Omega}_{\cI}$ and here, $\hat{\Omega}_{\cI}=\Omega_{M}$.

In Heisenberg picture, the quantum instrument $\cI$ transforms a measurement $N=\{N(y)\}\in\mathscr{M}(\cK)$ to another measurement 

\begin{align}
    \cI^{\dagger}[N]=\{\Phi^{\dagger}_x[N(y)]\}_{(x,y)\in (\hat{\Omega}_{\cI}\times\Omega_M)}.
\end{align}

Clearly, the measurement $\cI^{\dagger}[N]$ has the outcome set $\Omega_{\cI^{\dagger}[N]}=(\hat{\Omega}_{\cI}\times\Omega_N)$. A general transformation of a measurement can be the compositions of postprocessing of outcomes, probabilistic mixing with another measurement and action of quantum instruments. Similar to the sets of quantum channels, transformations of a set of measurements is slightly more general. Such transformations allow the controlled implemention of measurements in that set, in general.

 Given an arbitrary measurement $M=\{M(x)\}\in\mathscr{M}(\cH_A)$, one can associate it with a measure-prepare channel $\Gamma_M$ such that
\begin{align}
    \Gamma_{M}(\rho)=\sum_{a}\tr[\rho M(a)]\ket{a}\bra{a}
\end{align}
for all $\rho\in\cL(\cH_A)$ where $\{\ket{a}\}$ is a chosen orthonormal basis of Hilbert space $\cH_{\Omega_M}$ with dimension $|\Omega_M |$  where we denote the cardinality of a set by the symbol $|.|$.  Clearly, $\Gamma_{M}\in\mathscr{C}(\cH_A,\cH_{\Omega_M})$ and $\Gamma_{M\otimes N}=\Gamma_M\otimes\Gamma_N$ for two arbitrary measurements $M\in\mathscr{M}(\cH_A)$ and $N\in\mathscr{M}(\cH_B)$.


Now, we can define a trivially enlarged version of $M$ as $\hat{M}_{\cH_A\otimes\cH_B}=\{\hat{M}_{\cH_A\otimes\cH_B}(a)=M(a)\otimes\Id_{\cH_B}\}\in\mathscr{M}(\cH_A\otimes\cH_B)$. Clearly, we have $\Gamma_{\hat{M}_{\cH_A\otimes\cH_B}}=\Gamma_{M}\otimes\tr_{\cH_B}\in\mathscr{C}(\cH_A\otimes\cH_B,\cH_{\Omega_M})$. Performing the measurement $\hat{M}_{\cH_A\otimes\cH_B}$ on a system with an arbitrary quantum state $\rho_{AB}\in\cS(\cH_A\otimes\cH_B)$ is equivalent to performing the measurement $M$ on the $A$ part and doing nothing on the $B$ part. In other words, $\hat{M}_{\cH_A\otimes\cH_B}=M\otimes\cT_{\cH_B}$. Such trivial enlargements will be used repeatedly in Sec. \ref{Subsec:dist_based_res_meas}.

Consider the quantum channel $\tr_{\cH_{\Omega_N}}$ for a measurement $N\in\mathscr{M}(\cH_B)$. Then

\begin{align}
&(\mathbbm{I}_{\cH_{\Omega_M}}\otimes\tr_{\cH_{\Omega_N}})\circ(\Gamma_{M}\otimes\Gamma_N)(\rho_{AB})\nonumber\\
=&\sum_{ab}\tr[\rho_{AB}(M(a)\otimes N(b))]\ket{a}\bra{a}\nonumber\\
=&\sum_{a}\tr[\rho_{AB}(M(a)\otimes \sum_bN(b))]\ket{a}\bra{a}\nonumber\\
=&\sum_{a}\tr[\rho_{AB}(M(a)\otimes \Id_{\cH_B})]\ket{a}\bra{a}\nonumber\\
=&\Gamma_{\hat{M}_{\cH_A\otimes\cH_B}}(\rho_{AB}).
\end{align}
 Therefore, $\tr_{\cH_{\Omega_N}}$ is equivalent to the marginalization operation on the measurement $N$. Clearly, for an arbitrary measurement $M=\{M(m_1,m_2)\}\in\mathscr{M}(\cH)$ with outcome set $\Omega_A=X\times Y$ where the set $X$ and $Y$ are some sets of natural numbers and $m_1\in X$ and $m_2\in Y$ we can define a channel $\Gamma_{M}:\cL(\cH)\rightarrow\cL(\cH_X\otimes\cH_Y)$ with $\dim{\cH_X}=|X|$ and $\dim{\cH_Y}=|Y|$. Moreover, we denote the \emph{merginalization process} as $\mathbf{Merg}_{Y}$. In other words, $\mathbf{Merg}_{Y}(M)=\{[\mathbf{Merg}_{Y}(M)](m_1)=\sum_{m_2} M(m_1,m_2)\}$. Clearly,
 \begin{align}
     \Gamma_{\mathbf{Merg}_{Y}(M)}=(\mathbbm{I}_{\cH_{X}}\otimes\tr_{\cH_{Y}})\circ\Gamma_M.
 \end{align}
 
Now, given a pair of measurements $M_1$ and $M_2$, one can define a distance between them. One possible distance measure between $M_1\in\mathscr{M}(\cH)$ and $M_2\in\mathscr{M}(\cH)$ is defined as \cite{Tendick_dist_res_meas}

\begin{align}
    \cD_{\Diamond}(M_1, M_2)&:=\cD_{\Diamond}(\Gamma_{M_1}, \Gamma_{M_2})\\
    &=\mid\mid \Gamma_{M_1}-\Gamma_{M_2}\mid\mid_{\Diamond}.
\end{align}
Clearly, $\cD_{\Diamond}(M_1, M_2)$ is one of the possible choices of distance measures \cite{Tendick_dist_res_meas}. 

\begin{remark}
   \rm{ Note that without loss of generality the outcome sets $\Omega_{M_1}$ and $\Omega_{M_2}$ can always be assumed to be same by choosing some extra POVM elements to be zero matrix. More specifically, if $M_1$ has more number outcomes, a number of zero matrix can always be appended in $M_2$ as extra POVM elements so that the outcome sets $\Omega_{M_1}$ and $\Omega_{M_2}$ become equal to each other.}
\end{remark}



Consider two arbitrary measurements $M=\{M(a_1,a_2)\}\in\mathscr{M}(\cH)$ and $N=\{N(b_1,b_2)\}\in\mathscr{M}(\cH)$ with $\Omega_M=\Omega_N=X\times Y$. Clearly,

\begin{align}
    \cD_{\Diamond}(\mathbf{Merg}_{Y}(M), \mathbf{Merg}_{Y}(N))=&\cD_{\Diamond}(\Gamma_{\mathbf{Merg}_{Y}(M)}, \Gamma_{\mathbf{Merg}_{Y}(N)})\nonumber\\
    =&\mid\mid\Gamma_{\mathbf{Merg}_{Y}(M)}-\Gamma_{\mathbf{Merg}_{Y}(N)}\mid\mid_{\Diamond}\nonumber\\
    =& \mid\mid(\mathbbm{I}_{\cH_{X}}\otimes\tr_{\cH_{Y}})\circ[\Gamma_{M}-\Gamma_{N}]\mid\mid_{\Diamond}\nonumber\\
    \leq&\mid\mid\Gamma_{M}-\Gamma_{N}\mid\mid_{\Diamond}\nonumber\\
    =&\cD_{\Diamond}(M, N).\label{Eq:diam_non_inc_merg}
\end{align}
That is, the distance $\cD_{\Diamond}(M, N)$ is contractive under the marginalization operation.

 In Heisenberg picture, a $\mathbf{SWAP}$ channel acts on an arbitrary measurement $M_B\otimes M_A\in\mathscr{M}(\cH
_B\otimes\cH_A)$ and provides the outcome $\mathbf{SWAP}^{\dagger}_{\cH_A\leftrightarrow \cH_B}(M_B\otimes M_A)=M_A\otimes M_B\in\mathscr{M}(\cH_A\otimes\cH_B)$. Note that 
\begin{align}
    \Gamma_{M_A\otimes M_B}=\mathbf{SWAP}_{\cH_{\Omega_{M_B}}\leftrightarrow \cH_{\Omega_{M_A}}}\circ\Gamma_{M_B\otimes M_A}\circ\mathbf{SWAP}_{\cH_A\leftrightarrow \cH_B}.\label{Eq:gamma_ten_swap}
\end{align}
Consider arbitrary measurements $M_{CAB},N_{CAB}\in\mathscr{M}(\cH_C\otimes\cH_A\otimes\cH_B)$. Then from Eq. \eqref{Eq:swap_inv_dimon}, and Eq. \eqref{Eq:gamma_ten_swap} we obtain

\begin{align}
    \cD_{\Diamond}(M_{CAB},N_{CAB})=\cD_{\Diamond}(M_{CBA},N_{CBA}),\label{Eq:swap_inv_dimon_meas}
\end{align}
where $M_{CBA}=\mathbbm{I}^{\dagger}_{\cH_C}\otimes\mathbf{SWAP}^{\dagger}_{\cH_B\leftrightarrow \cH_A}(M_{CAB})$ and $N_{CBA}=\mathbbm{I}^{\dagger}_{\cH_C}\otimes\mathbf{SWAP}^{\dagger}_{\cH_B\leftrightarrow \cH_A}(N_{CAB})$ .
In Section \ref{Subsec:Gen_t_set_meas}, we will define the distance for two sets of measurements.

\subsection{Aspects of a resource measure based on  set of quantum measurements}
\label{Subsec:prelim_resource}
Two major ingredients of a generic resource theory concerning a specific resource are the free objects and the free transformations (also known as free operations for the state-based resource theories). Free objects are the objects that do not contain that particular resource and free transformations are the transformations that transform a free object to another free object. The objects can be quantum states, quantum measurements, quantum channels etc. depending on a particular resource theory and similarly, transformations can be quantum channels, quantum super-channels etc depending on a particular resource theory. A resource theory is said to be \emph{convex} if free objects and free transformations form convex sets.

Consider a resource theory with the set of free objects $\cF$ and the set of free transformations $\cO$. Without loss of generality, we can assume $\cF$ and $\cO$ are compact. In the context of measurement-based resource theories, the objects are measurements or more generally sets of measurements and similarly, transformations are the transformations of one measurement to another measurement or transformations from one set of measurements to another set of measurements in Heisenberg picture.

Certain natural assumptions are respected by any state based resource theories as defined in  Ref. \cite{choi_defined,Chitambar_QRT_review}. Similarly, in case of any measurements-based resource theory, one can make the following natural assumptions-

\begin{enumerate}
    \item[A1.] Two objects $\cM_1$ and $\cM_2$ are free if and only if $\cM_1\otimes\cM_2$ is free.
    \item[A2.] Two transformations $\cW_1$ and $\cW_2$ are free if and only if  $\cW_1\otimes\cW_2$ is free.
    \item[A3.] As in Heisenberg picture, identity channel does not transform the measurements (and therefore, does not transform any free object), it is a free transformation. 
    \item[A4] A free object remains a free object under the swapping of the Hilbert spaces.
    \item[A5.] $\cT_{\cH}$ is a free object for all $\cH$.
    \item[A6.] Merginalization process is a free transformation.
\end{enumerate}


It should be noted that all three known measurement-based resource theories respect the natural assumptions [A1]–[A6] and are expected to be respected by any measurement-based resource theories to be developed in the future. Merginalization is a free transformation for the resource theory of measurement incompatibility and the resource theory of measurement coherence. Now, it is known that the post-processing of outcomes is not a free transformation for the resource theory of measurement sharpness, in general \cite{Mitra_meas_sharp,Buscemi_meas_sharp}. However, merginalization is very special case of post-processing of outcomes that transforms an arbitrary measurement $M\in\mathscr{M}(\cH)$ to $\cT_{\cH}$ which is a trivial object. Therefore, technically, the merginalization should be the free transformation for the the resource theory of measurement sharpness. But in Ref. \cite{Buscemi_meas_sharp},  only the outcome set preserving free transformations for the resource theory of measurement sharpness have been considered. The assumption [A6] will not be used everywhere in the manuscript. We will explicitly mention wherever we use the assumption [A6].

 Now, the quantification of resources or the resource measures are integral part of any resource theory. A resource measure $\mathbbm{R}$ should satisfy the following conditions-

\begin{enumerate}
    \item[R1.] (Non-negativity and faithfulness): $\mathbbm{R}(\cM)\geq 0$ for all objects $\cM$ and $\mathbbm{R}(\cM)= 0$ if and only if $\cM\in\cF$ \label{Con:R1}.
    \item[R2.] (Monotonicity): $\mathbbm{R}(\cM)\geq \mathbbm{R}(\cW(\cM))$ for all objects $\cM$ and for all free transformations $\cW$ i.e., for all $\cW\in\cO$ \label{Con:R2}.
\end{enumerate}

In addition to these necessary conditions, one more desirable condition is

\begin{enumerate}
    \item[R3.] (Convexity): $\mathbbm{R}(\sum^n_i p_i \cM_i)\leq \sum^n_i p_i\mathbbm{R}( \cM_i)$ for all sets of objects $\{\cM_i\}^n_{i=1}$ and all probability distributions $\{p_i\}^n_{i=1}$ \label{Con:R3}.
\end{enumerate}

A resource measure that satisfies these three conditions is usually said to be a good measure to work with for any convex resource theory. From now on, we assume resource measures to be a continuous function.  In our scenario, the objects are the sets of measurements and transformations are the physically realizable maps that transform a set of measurements to another set of measurements in Heisenberg picture. Two well-known resource measures that will be used later in this work are the resource robustness and the resource weight, as defined below.

For convex measurement-based resource theories, the resource robustness of an arbitrary set of measurements $\cM=\{M_i\}\subset\mathscr{M}(\cH_A)$ is defined as
\begin{equation}
\begin{split}
\mathscr{R}(\cM)=\min~  & \left. r\right. \\
\text{s.t.} ~ & \left. \cN=\{N_i=\frac{M_i}{1+r}+ \frac{r\tilde{M}_i}{1+r}\}\in\cF_{\cH_A} \right. \\
& \left. \tilde{\cM}=\{\tilde{M}_i\}\subset\mathscr{M}(\cH_A). \right.
\end{split}\label{Eq:def_robust}
\end{equation}
Here, $\tilde{\cM}$ is an arbitrary set of unwanted noise measurements, $\cF_{\cH_A}$ is the set of free sets of measurements acting on Hilbert space $\cH_A$  and the optimization is over all variables, other than the given set of measurements $\cM$.

Similarly, for convex measurement-based resource theories, the resource weight of an arbitrary set of measurements $\cM=\{M_i\}\subset\mathscr{M}(\cH_A)$ is defined as
\begin{equation}
\begin{split}
\mathscr{W}(\cM)=\min~  & \left. r\right. \\
\text{s.t.} ~ & \left. \cM=\{M_i=\frac{\tilde{N}_i}{1+r}+ \frac{r\tilde{M}_i}{1+r}\} \right. \\
& \left. \tilde{\cM}=\{\tilde{M}_i\}\subset\mathscr{M}(\cH_A) \right. \\
& \left. \tilde{\cN}=\{\tilde{N}_i\}\in\cF_{\cH_A}. \right.
\end{split}\label{Eq:def_weight}
\end{equation}
Here, $\tilde{\cM}$ is an arbitrary set of unwanted noise measurements, $\cF_{\cH_A}$ is the set of free sets of measurements acting on Hilbert space $\cH_A$ and the optimization is over all variables, other than the given set of measurements $\cM$.


Given an arbitrary resource measure $\mathbbm{R}$, one can define a family of resource measures that depend on the parameter $\epsilon$. In the following, we write down the definition of such a measure.

\begin{definition}
    For any resource measure $\mathbbm{R}$ in any measurement-based resource theory, and for a generic distance $\mathbf{D}$, the corresponding $\epsilon$-measure $\mathbbm{R}^{\mathbf{D}}_{inf,\varepsilon}$ of a set of measurements $\cM=\{M_i\}\subset \mathscr{M}(\cH)$ is defined as

    \begin{align}
    \mathbbm{R}^{\mathbf{D}}_{inf,\epsilon}(\cM)=\inf_{\substack{\cN\subset \mathscr{M}(\cH),\\
    \mathbf{D}(\cM,\cN)\leq \epsilon,~|\cM|=|\cN|}}\mathbbm{R}(\cN). \label{Eq:def_eps_inf}
    \end{align}
\end{definition}
Note that as we restrict ourselves to finite dimensional Hilbert space, the set $\{\cB\mid \mathbf{D}(\cA,\cB)\leq \epsilon\}$ is convex and compact. One of the goals of our work is to study the properties of $\epsilon$-measures of measurement-based resources.

\section{Main results}
\label{Sec:Main_results}
\subsection{Transformation of sets of measurements, sets of channels and a distance measure}
\label{Subsec:Gen_t_set_meas}
As we already mentioned in Sec. \ref{Sec:Prelim}, in this Section, we discuss fairly general transformation of a set of channels. 

Consider a set of quantum channels $\cC=\{\Lambda_i\in\mathscr{C}(\cH,\cK)\}$. Throughout the paper, we say that the set $\cC$ is equal to another set $\cC^{\prime}=\{\Lambda^{\prime}_i\in\mathscr{C}(\cH,\cK)\}$ if and only if $\Lambda_i=\Lambda^{\prime}_i$ holds for all $i$s. In other words, we assume that all elements of the $\cC$ is rigidly marked by the index $i$ throughout the paper. Similar assumption will be made for sets of measurements throughout the paper. From now on, we restrict ourselves to the sets of channels and sets of measurements with finite number of elements. Then a fairly general transformation $\overline{\cV}$ can be written as

\begin{align}
    [\overline{\cV}(\cC)]_j=\overline{\Theta}_{post}^j\circ(\Sigma_{\cC}\otimes \mathbbm{I_{\overline{\cR}}})\circ\overline{\Theta}_{pre}^j,\label{eq:far_gen_tra}
\end{align}

where $[\overline{\cV}(\cC)]_j$ is the $j$th channel of the set $\overline{\cV}(\cC)$, $\Sigma_{\cC}=\sum_{\overline{i}}\Lambda_{\overline{i}}\otimes\overline{\Phi}_{\overline{i}}$, $\overline{\Phi}_{\overline{i}}(.)=\ket{\overline{i}}\bra{\overline{i}}(.)\ket{\overline{i}}\bra{\overline{i}}$ for all $\overline{i}$, $\overline{\Theta}_{pre}^j\in\mathscr{C}({\overline{\cH},\cH\otimes\cH_I\otimes\overline{\cR}})$ for all $j$, $\overline{\Theta}_{post}^j\in\mathscr{C}({\cK\otimes\cH_I\otimes\overline{\cR}, \overline{\cK}})$ for all $j$, $\{\ket{\overline{i}}\}$ is an orthonormal basis of Hilbert space $\cH_I$. Note that $\Sigma_{\cC}$ is the \emph{controlled implementation} of the channels in the set $\cC$. Clearly, if both $\cC$ and $\cV[\cC]$ contain one quantum channel then \eqref{eq:far_gen_tra} reduces to \eqref{Eq:sup_chan} as a special case.

\begin{center}
    \begin{figure}[hbt!]
     \includegraphics[ width=8.9cm, height=4.3cm]{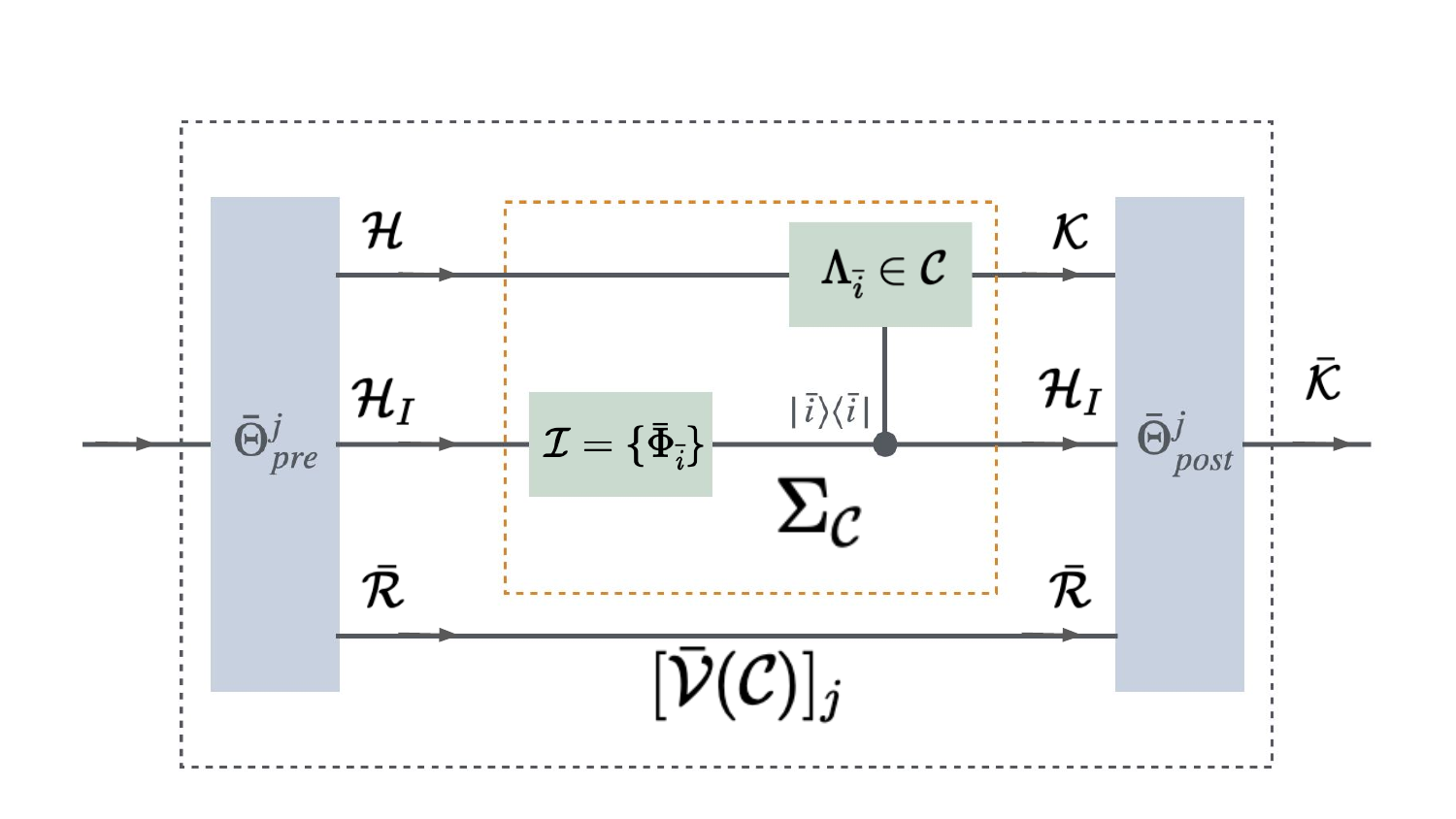}
    \caption{Schematic circuit diagram of the fairly general transformation of a set of channels $\cC$ to another set of channels $\overline{\cV}(\cC)$ as represented through Eq. \eqref{eq:far_gen_tra}. In this figure, instead of the actual states of the input and output systems, we mostly indicate input and output Hilbert spaces because it is more important for our purpose of illustration.}
    \label{fig:CFEWF}
\end{figure}
\end{center}

Now, we consider a function of two sets of channels $\cC_1=\{\Lambda_i\in\mathscr{C}(\cH,\cK_i)\}^n_{i=1}$ and $\cC_2=\{\Theta_i\in\mathscr{C}(\cH,\cK_i)\}^n_{i=1}$ as

\begin{align}
    \overline{\cD}(\cC_1,\cC_2)=\max_{i\in\{1,\ldots,n\}}\cD_{\Diamond}(\Lambda_i,\Theta_i).\label{eq:dist_set_chan}
\end{align}
Clearly, if each $\cC_1$ and $\cC_2$ contain just one quantum channel then $\overline{\cD}$ coincides with $\cD_{\Diamond}$. Now, we show that $\overline{\cD}(\cC_1,\cC_2)$ is a valid distance measure.

\begin{lemma}
    For two sets of channels $\overline{\cD}(\cC_1,\cC_2)$ is a valid distance.
\end{lemma}

\begin{proof}
    Firstly, as the diamond distance $\cD_{\Diamond}$ is always positive-semidefinite for an arbitrary pair of quantum channels, from Eq. \eqref{eq:dist_set_chan} we obtain $\overline{\cD}(\cC_1,\cC_2)\geq 0$. Now, note that $\cD_{\Diamond}(\Lambda_i,\Theta_i)\geq 0$ for all $i$s. Therefore, if $\overline{\cD}(\cC_1,\cC_2)= 0$ then from Eq. \eqref{eq:dist_set_chan}, we obtain that $\cD_{\Diamond}(\Lambda_i,\Theta_i)=0$ for all $i$s which implies $\Lambda_i=\Theta_i$ for all $i$s or in other words $\cC_1=\cC_2$. Conversely, if $\Lambda_i=\Theta_i$ for all $i$s then $\cD_{\Diamond}(\Lambda_i,\Theta_i)=0$ for all $i$s which implies $\overline{\cD}(\cC_1,\cC_2)=0$.

    Secondly, as $\cD_{\Diamond}(\Lambda_i,\Theta_i)=\cD_{\Diamond}(\Theta_i,\Lambda_i)$ for all $i$s, $\overline{\cD}(\cC_1,\cC_2)=\overline{\cD}_{\Diamond}(\cC_2,\cC_1)$.

    Thirdly, consider another set of channels $\cC_3=\{\Psi_i\}$. Now, suppose that in Eq. \eqref{eq:dist_set_chan} supremum occurs for $i=i^*$ then

    \begin{align}
        \overline{\cD}(\cC_1,\cC_2)&=\cD_{\Diamond}(\Lambda_{i^*},\Theta_{i^*})\nonumber\\
        &\leq  \cD_{\Diamond}(\Lambda_{i^*},\Psi_{i^*})+\cD_{\Diamond}(\Psi_{i^*},\Theta_{i^*})\nonumber\\
        &\leq \max_i\cD_{\Diamond}(\Lambda_{i},\Psi_{i})+\max_i\cD_{\Diamond}(\Psi_{i},\Theta_{i})\nonumber\\
        &=\overline{\cD}_{\Diamond}(\cC_1,\cC_3)+\overline{\cD}_{\Diamond}(\cC_3,\cC_2).
    \end{align}
    Hence, $\overline{\cD}$ satisfies triangular inequality.
    Therefore, $\overline{\cD}$ is a distance measure.
\end{proof}
 The tensor product of two sets of channels $\tilde{\cC}_1=\{\tilde{\Lambda}_i\in\mathscr{C}(\cH_{1},\cK_1)\}^m_{i=1}$ and $\tilde{\cC}_2=\{\tilde{\Theta}_j\in\mathscr{C}(\cH_{2},\cK_2)\}^n_{j=1}$ is defined as $\tilde{\cC}_1\otimes\tilde{\cC}_2:=\{\tilde{\Psi}_{ij}=\tilde{\Lambda}_i\otimes\tilde{\Theta}_j\}$. 
The convex combination of two sets of channels $\cC_1=\{\Lambda_i\in\mathscr{C}(\cH,\cK)\}^n_{i=1}$ and $\cC_2=\{\Theta_i\in\mathscr{C}(\cH,\cK)\}^n_{i=1}$ is defined as $p\cC_1+(1-p)\cC_2:=\{\Psi_i=p\Lambda_i+(1-p)\Theta_i\}$. We denote a process where a given set of channels $\cC_1$ (as an input)  is probabilistically mixed with an arbitrary fixed set of channels $\cQ=\{\Lambda^{\prime}_i\}$ as

\begin{align}
    \cP_{p,\cQ}[\cC_1]:=p\cC_1+(1-p)\cQ,
\end{align}
where $0\leq p\leq 1$ is the probability.

Now, we start studying different mathematical properties of the distance $\overline{\cD}$.
\begin{proposition}
$\overline{\cD}$ satisfies the following relations
\begin{enumerate}
\item $\overline{\cD}$ satisfies joint convexity property. In other words, $\overline{\cD}(p\cC_1+(1-p)\cC_2,p\cC^{\prime}_1+(1-p)\cC^{\prime}_2)\leq p\overline{\cD}(\cC_1,\cC^{\prime}_1)+(1-p)\overline{\cD}(\cC_2,\cC^{\prime}_2)$.
\item $\overline{\cD}$ is jointly subadditive under tensor product. In other words, $\overline{\cD}(\tilde{\cC}_1\otimes\tilde{\cC}_2, \tilde{\cC}^{\prime}_1\otimes\tilde{\cC}^{\prime}_2)\leq\overline{\cD}(\tilde{\cC}_1, \tilde{\cC}^{\prime}_1)+\overline{\cD}(\tilde{\cC}_2,\tilde{\cC}^{\prime}_2).$ If $\tilde{\cC}_2=\tilde{\cC}^{\prime}_2)$, equality holds i.e., $\overline{\cD}$ is invariant under the tensor product with same set of quantum channels.
\end{enumerate}
\label{Proposi:join_convex_subaddit_ten_product}
\end{proposition}
\begin{proof}
We use joint convexity of diamond distance and multiplicity of diamond norm to prove these statements.
\begin{enumerate}
\item Let $\overline{\cD}(p\cC_1+(1-p)\cC_2,p\cC^{\prime}_1+(1-p)\cC^{\prime}_2)=\cD_{\Diamond}(p\Lambda_{i^*}+(1-p)\Theta_{i^*},p\Lambda_{i^*}^{\prime}+(1-p)\Theta_{i^*}^{\prime})$. Then
\begin{align}
&\overline{\cD}(p\cC_1+(1-p)\cC_2,p\cC^{\prime}_1+(1-p)\cC^{\prime}_2)\nonumber\\
&=\cD_{\Diamond}(p\Lambda_{i^*}+(1-p)\Theta_{i^*},p\Lambda_{i^*}^{\prime}+(1-p)\Theta_{i^*}^{\prime})\nonumber\\
&\leq p\cD_{\Diamond}(\Lambda_{i^*},\Lambda_{i^*}^{\prime})+(1-p)\cD_{\Diamond}(\Theta_{i^*},\Theta_{i^*}^{\prime})\nonumber\\
&\leq p\max_i\cD_{\Diamond}(\Lambda_{i},\Lambda_{i}^{\prime})+(1-p)\max_i\cD_{\Diamond}(\Theta_{i},\Theta_{i}^{\prime})\nonumber\\
&=p\overline{\cD}(\cC_1,\cC_2)+(1-p)\overline{\cD}(\tilde{\cC}_1,\tilde{\cC}_2),
\end{align}

where in the third line, we have used the joint convexity property of $\cD_{\Diamond}$ given in Eq. \eqref{Eq:dim_dist_conv}.

\item Let $\overline{\cD}(\tilde{\cC}_1\otimes\tilde{\cC}_2, \tilde{\cC}^{\prime}_1\otimes\tilde{\cC}^{\prime}_2)=\cD_{\Diamond}(\tilde{\Lambda}_{i^*}\otimes\tilde{\Theta}_{j^*}, \tilde{\Lambda}^{\prime}_{i^*}\otimes\tilde{\Theta}^{\prime}_{j^*})$. Then

\begin{align}
&\overline{\cD}(\tilde{\cC}_1\otimes\tilde{\cC}_2, \tilde{\cC}^{\prime}_1\otimes\tilde{\cC}^{\prime}_2)\nonumber\\
&=\cD_{\Diamond}(\tilde{\Lambda}_{i^*}\otimes\tilde{\Theta}_{j^*}, \tilde{\Lambda}^{\prime}_{i^*}\otimes\tilde{\Theta}^{\prime}_{j^*})\nonumber\\
&\leq\cD_{\Diamond}(\tilde{\Lambda}_{i^*}\otimes\tilde{\Theta}_{j^*}, \tilde{\Lambda}^{\prime}_{i^*}\otimes\tilde{\Theta}_{j^*})+\cD_{\Diamond}(\tilde{\Lambda}^{\prime}_{i^*}\otimes\tilde{\Theta}_{j^*}, \tilde{\Lambda}^{\prime}_{i^*}\otimes\tilde{\Theta}^{\prime}_{j^*})\nonumber\\
&=\mid\mid\tilde{\Lambda}_{i^*}\otimes\tilde{\Theta}_{j^*}- \tilde{\Lambda}^{\prime}_{i^*}\otimes\tilde{\Theta}_{j^*}\mid\mid_{\Diamond}+\mid\mid\tilde{\Lambda}^{\prime}_{i^*}\otimes\tilde{\Theta}_{j^*}- \tilde{\Lambda}^{\prime}_{i^*}\otimes\tilde{\Theta}^{\prime}_{j^*}\mid\mid_{\Diamond}\nonumber\\
&=\mid\mid(\tilde{\Lambda}_{i^*}- \tilde{\Lambda}^{\prime}_{i^*})\otimes\tilde{\Theta}_{j^*}\mid\mid_{\Diamond}+\mid\mid\tilde{\Lambda}^{\prime}_{i^*}\otimes(\tilde{\Theta}_{j^*}- \tilde{\Theta}^{\prime}_{j^*})\mid\mid_{\Diamond}\nonumber\\
&=\mid\mid(\tilde{\Lambda}_{i^*}- \tilde{\Lambda}^{\prime}_{i^*})\mid\mid_{\Diamond}.\mid\mid\tilde{\Theta}_{j^*}\mid\mid_{\Diamond}+\mid\mid\tilde{\Lambda}^{\prime}_{i^*}\mid\mid_{\Diamond}.\mid\mid(\tilde{\Theta}_{j^*}- \tilde{\Theta}^{\prime}_{j^*})\mid\mid_{\Diamond}\nonumber\\
&=\mid\mid(\tilde{\Lambda}_{i^*}- \tilde{\Lambda}^{\prime}_{i^*})\mid\mid_{\Diamond}+\mid\mid(\tilde{\Theta}_{j^*}- \tilde{\Theta}^{\prime}_{j^*})\mid\mid_{\Diamond}\nonumber\\
&=\cD_{\Diamond}(\tilde{\Lambda}_{i^*}- \tilde{\Lambda}^{\prime}_{i^*})+\cD_{\Diamond}(\tilde{\Theta}_{j^*}- \tilde{\Theta}^{\prime}_{j^*})\nonumber\\
&\leq\max_i\cD_{\Diamond}(\tilde{\Lambda}_i- \tilde{\Lambda}^{\prime}_i)+\max_j\cD_{\Diamond}(\tilde{\Theta}_j- \tilde{\Theta}^{\prime}_j)\nonumber\\
&=\overline{\cD}(\tilde{\cC}_1, \tilde{\cC}^{\prime}_1)+\overline{\cD}(\tilde{\cC}_2,\tilde{\cC}^{\prime}_2),\label{Eq:subadditiv_ten_chan}
\end{align}

where sixth line, we have used the multiplicative property of diamond norm under tensor product (see Remark \ref{Re:Diam_norm}) and in the seventh line, we have used the fact that the diamond norm of a quantum channel is $1$ (see Remark \ref{Re:Diam_norm}).

Now, suppose $\tilde{\cC}_2=\tilde{\cC}^{\prime}_2$. Then From Eq. \eqref{Eq:subadditiv_ten_chan}, we obtain
\begin{align}
    \overline{\cD}(\tilde{\cC}_1\otimes\tilde{\cC}_2, \tilde{\cC}^{\prime}_1\otimes\tilde{\cC}_2)\leq\overline{\cD}(\tilde{\cC}_1, \tilde{\cC}^{\prime}_1).
\end{align}

Now, let
\begin{align}
\overline{\cD}(\tilde{\cC}_1, \tilde{\cC}^{\prime}_1)=&\cD_{\Diamond}(\tilde{\Lambda}_{i^*},\tilde{\Lambda}^{\prime}_{i^*})\nonumber\\
=&\mid\mid\tilde{\Lambda}_{i^*}-\tilde{\Lambda}^{\prime}_{i^*}\mid\mid_{\Diamond}\nonumber\\
=&\mid\mid(\tilde{\Lambda}_{i^*}-\tilde{\Lambda}^{\prime}_{i^*})\mid\mid_{\Diamond}.\mid\mid\Theta_j\mid\mid_{\Diamond}\nonumber\\
=&\mid\mid(\tilde{\Lambda}_{i^*}-\tilde{\Lambda}^{\prime}_{i^*})\otimes\Theta_j\mid\mid_{\Diamond}\nonumber\\
\leq&\max_{i,j}\mid\mid(\tilde{\Lambda}_{i}\otimes\Theta_j-\tilde{\Lambda}^{\prime}_{i}\otimes\Theta_j)\mid\mid_{\Diamond}\nonumber\\
=&\overline{\cD}(\tilde{\cC}_1\otimes\tilde{\cC}_2, \tilde{\cC}^{\prime}_1\otimes\tilde{\cC}_2).
\end{align}
Hence, 
\begin{align}
    \overline{\cD}(\tilde{\cC}_1\otimes\tilde{\cC}_2, \tilde{\cC}^{\prime}_1\otimes\tilde{\cC}_2)=\overline{\cD}(\tilde{\cC}_1, \tilde{\cC}^{\prime}_1).
\end{align}
\end{enumerate}
\end{proof}

\begin{corollary}
    $\overline{\cD}(\cP_{p,\cQ}(\cC_1),\cP_{p,\cQ}(\cC_2))\leq\overline{\cD}(\cC_1,\cC_2)$. In other words, $\overline{\cD}$ is non-increasing under the probabilistic mixing with a fixed set of quantum channels\label{Coro:mono_set_prob_mix}.
\end{corollary}

\begin{proof}
     The result is immediate from Proposition \ref{Proposi:join_convex_subaddit_ten_product}.
     

    
\end{proof}

Now, we show that $\overline{\cD}$ is invariant under the swapping of Hilbert spaces.

\begin{lemma}
  For arbitrary $\cC_{CAB}=\{\Lambda_i\in\mathscr{C}(\cH_C\otimes\cH_A\otimes\cH_B)\}$ and $\cC^{\prime}_{CAB}=\{\Lambda^{\prime}_j\in\mathscr{C}(\cH_C\otimes\cH_A\otimes\cH_B)\}$

\begin{align}
    \overline{\cD}(\cC_{CAB},\cC^{\prime}_{CAB})=\overline{\cD}(\cC_{CBA},\cC^{\prime}_{CBA}),
\end{align}
where $\cC_{CBA}=\{\Theta_i=(\mathbbm{I}_{\cH_C}\otimes\mathbf{SWAP}_{\cK_A\leftrightarrow \cK_B})\circ\Lambda_i\circ(\mathbbm{I}_{\cH_C}\otimes\mathbf{SWAP}_{\cH_B\leftrightarrow \cH_A})\}$ and $\cC_{CBA}=\{\Theta^{\prime}_j=\mathbbm{I}_{\cH_C}\otimes\mathbf{SWAP}_{\cK_A\leftrightarrow \cK_B})\circ\Lambda^{\prime}_j\circ(\mathbbm{I}_{\cH_C}\otimes\mathbf{SWAP}_{\cH_B\leftrightarrow \cH_A})\}$. In other words, $\overline{\cD}$ is invariant under swapping of Hilbert spaces. \label{le:d_bar_inv_swap}
\end{lemma}

\begin{proof}
    Let $\overline{\cD}(\cC_{CAB},\cC^{\prime}_{CAB})=\cD_{\Diamond}(\Lambda_{i^*},\Lambda^{\prime}_{i^*})$. Then
    \begin{align}
        \overline{\cD}(\cC_{CAB},\cC^{\prime}_{CAB})=&\cD_{\Diamond}(\Lambda_{i^*},\Lambda^{\prime}_{i^*})\nonumber\\
        =&\cD_{\Diamond}(\Theta_{i^*},\Theta^{\prime}_{i^*})\nonumber\\
        \leq&\max_i\cD_{\Diamond}(\Theta_{i},\Theta^{\prime}_{i})\nonumber\\
        =&\overline{\cD}(\cC_{CBA},\cC^{\prime}_{CBA}),
    \end{align}
    where we have used Eq. \eqref{Eq:swap_inv_dimon} in the second inequality.
    Similarly,
    \begin{align}
        \overline{\cD}(\cC_{CBA},\cC^{\prime}_{CBA})=&\cD_{\Diamond}(\Theta_{i^*},\Theta^{\prime}_{i^*})\nonumber\\
        =&\cD_{\Diamond}(\Lambda_{i^*},\Lambda^{\prime}_{i^*})\nonumber\\
        \leq&\max_i\cD_{\Diamond}(\Lambda_{i},\Lambda^{\prime}_{i})\nonumber\\
        =&\overline{\cD}(\cC_{CAB},\cC^{\prime}_{CAB}).
    \end{align}
    Hence, the lemma is proved.

\end{proof}




Now, we prove a crucial property of $\overline{\cD}$ that will be implicitly used in later parts of our work.

\begin{theorem}\label{theor:mono_set_tran}
    $\overline{\cD}(\overline{\cV}(\cC_1),\overline{\cV}(\cC_2))\leq\overline{\cD}(\cC_1,\cC_2)$ for any $\overline{\cV}$ of the form given in Eq. \eqref{eq:far_gen_tra}. In other words, $\overline{\cD}$ is contractive under the transformations of the form given in Eq. \eqref{eq:far_gen_tra}.
\end{theorem}

\begin{proof}
    Suppose, 
    \begin{align}
        \overline{\cD}(\overline{\cV}(\cC_1),\overline{\cV}(\cC_2))=\cD_{\Diamond}([\overline{\cV}(\cC_1)]_{j^*},[\overline{\cV}(\cC_2)]_{j^*}).
    \end{align}

    Then
    \begin{align}
        &\overline{\cD}(\overline{\cV}(\cC_1),\overline{\cV}(\cC_2)\nonumber\\
        =&\cD_{\Diamond}([\overline{\cV}(\cC_1)]_{j^*},[\overline{\cV}(\cC_2)]_{j^*})\nonumber\\
        =&\mid\mid[\overline{\cV}(\cC_1)]_{j^*}-[\overline{\cV}(\cC_2)]_{j^*}\mid\mid_{\Diamond}\nonumber\\
        =&\mid\mid\overline{\Theta}_{post}^{j^*}\circ(\Sigma_{\cC_1}\otimes \mathbbm{I_{\overline{\cR}}})\circ\overline{\Theta}_{pre}^{j^*}-\overline{\Theta}_{post}^{j^*}\circ(\Sigma_{\cC_2}\otimes \mathbbm{I_{\overline{\cR}}})\circ\overline{\Theta}_{pre}^{j^*}\mid\mid_{\Diamond}\nonumber\\
        =&\mid\mid\overline{\Theta}_{post}^{j^*}\circ(\Sigma_{\cC_1}\otimes \mathbbm{I_{\overline{\cR}}}-\Sigma_{\cC_2}\otimes \mathbbm{I_{\overline{\cR}}})\circ\overline{\Theta}_{pre}^{j^*}\mid\mid_{\Diamond}\nonumber\\
        \leq&\mid\mid(\Sigma_{\cC_1}\otimes \mathbbm{I_{\overline{\cR}}}-\Sigma_{\cC_2}\otimes \mathbbm{I_{\overline{\cR}}})\mid\mid_{\Diamond}\nonumber\\
        =&\mid\mid(\Sigma_{\cC_1}-\Sigma_{\cC_2})\otimes \mathbbm{I_{\overline{\cR}}}\mid\mid_{\Diamond}\nonumber\\
        =&\mid\mid\Sigma_{\cC_1}-\Sigma_{\cC_2}\mid\mid_{\Diamond}\nonumber\\
        =&\max_{\overline{\rho}_{\cH\cH_I\cR}\in\cS(\cH\otimes\cH_I\otimes\cR)}\mid\mid(\Sigma_{\cC_1}-\Sigma_{\cC_2})\otimes\mathbbm{I_{\cR}}((\overline{\rho}_{\cH\cH_I\cR})\mid\mid_1 \label{Eq:mono_opt_diam}
    \end{align}
Clearly, from Eq. \eqref{Eq:Dia_formula_H_p}, we have $\dim(\cR)=\dim(\cH\otimes\cH_I)$.
    
Now, suppose maximum in Eq. \eqref{Eq:mono_opt_diam} occurs for $\overline{\rho}_{\cH\cH_I\cR}=\rho_{\cH\cH_I\cR}$. Then
    \begin{align}
        \overline{\cD}(\overline{\cV}(\cC_1),\overline{\cV}(\cC_2))\leq&\mid\mid(\Sigma_{\cC_1}-\Sigma_{\cC_2})\otimes\mathbbm{I}_{\cR}(\rho_{\cH\cH_I\cR})\mid\mid_1\nonumber\\
        =&\mid\mid\sum_i(\Lambda_i-\Theta_i)\otimes\mathbbm{I}_{\cR}(\bra{i}\rho_{\cH\cH_I\cR}\ket{i}\ket{i}\bra{i})\mid\mid_1\nonumber\\
        =&\sum_i\mid\mid(\Lambda_i-\Theta_i)\otimes\mathbbm{I}_{\cR}(\bra{i}\rho_{\cH\cH_I\cR}\ket{i})\mid\mid_1\nonumber\\
        =&\sum_i\mid\mid p_i(\Lambda_i-\Theta_i)\otimes\mathbbm{I}_{\cR}(\sigma_{\cH\cR})\mid\mid_1\nonumber\\
        =&\sum_ip_i\mid\mid(\Lambda_i-\Theta_i)\otimes\mathbbm{I}_{\cR}(\sigma_{\cH\cR})\mid\mid_1\nonumber\\
        \leq&\max_i\max_{\sigma^{\prime}_{\cH\cR}}\mid\mid(\Lambda_i-\Theta_i)\otimes\mathbbm{I}_{\cR}(\sigma^{\prime}_{\cH\cR})\mid\mid_1\nonumber\\
        =&\max_i\mid\mid(\Lambda_i-\Theta_i)\mid\mid_{\Diamond}\nonumber\\
        =&\max_i\cD_{\Diamond}(\Lambda_i,\Theta_i)\nonumber\\
        =&\overline{\cD}(\cC_1,\cC_2),
    \end{align}
    where in the fourth line $p_i=\tr[\bra{i}\rho_{\cH\cH_I\cR}\ket{i}]$, $\sigma_{\cH\cR}=\frac{\bra{i}\rho_{\cH\cH_I\cR}\ket{i}}{\tr[\bra{i}\rho_{\cH\cH_I\cR}\ket{i}]}\in\cS(\cH\otimes\cR)$ and in the seventh line, we have used the fact $\dim(\cR)=\dim(\cH\otimes\cH_I)=\dim(\cH).\dim(\cH_I)\geq\dim(\cH)$ and Remark \ref{Re:Diam_norm}.
\end{proof}

Now, we define the distance between two sets of measurements $\cM=\{M_i\}$ and $\cN=\{N_i\}$ as

\begin{align}
    \widetilde{\cD}(\cM,\cN):=&\overline{\cD}(\cG_{\cM},\cG_{\cN})\nonumber\\
    =&\max_{i\in\{1,\ldots,n\}}\cD_{\Diamond}(\Gamma_{M_i},\Gamma_{N_i})\label{Eq:def_dist_meas_of_set_meas},
\end{align}
where $\cG_{\cM}:=\{\Gamma_{M_i}\}$ and $\cG_{\cN}:=\{\Gamma_{N_i}\}$.




\begin{remark}
    \rm{It should be mentioned here that in Ref. \cite{Tendick_dist_res_meas}, authors studied the distance measure for a set of measurements and they defined the distance between two sets of measurements $\cM=\{M_i\}$ and $\cN=\{N_i\}$ as $\mathbf{D}_{\Diamond}(\cM,\cN)=\sum_ip_i\cD_{\Diamond}(\Gamma_{M_i},\Gamma_{N_i})$, where $\{p_i\}$ is an arbitrarily chosen probability distribution. Furthermore, the monotonicity of the resource measure (constructed from their distance measure) under the free transformation of the resource theory of incompatibility of measurements (given in Ref. \cite{Buscemi_meas_incomp}) has not been shown yet.} Instead of using that measure, in this work, we use the distance measure $\widetilde{\cD}$ which is independent of any such probability distribution and is a function of only the sets $\cM$ and $\cN$. Furthermore, neither $\mathbf{D}_{\Diamond}$ in Ref. \cite{Tendick_dist_res_meas} nor $\widetilde{\cD}$ is independent of permutation of the measurements in the sets. However, that does not affect any of our results and the resource measure created from $\mathbbm{D}$ (which is more general than $\widetilde{\cD}$) is independent of the permutation of the measurements as observed from Eq. \eqref{Eq:def_dist_res_meas}.
\end{remark}

Now, we can immediately write the following corollary from Proposition \ref{Proposi:join_convex_subaddit_ten_product}, Lemma \ref{le:d_bar_inv_swap}, and Corollary \ref{Coro:mono_set_prob_mix}.

\begin{corollary} 
The distance $\widetilde{\cD}$ have following Properties-
    \begin{enumerate}
        \item $\widetilde{\cD}$ satisfies joint convexity property.
        \item $\widetilde{\cD}$ is jointly subadditive under tensor product and is invariant under the tensor product with same set of quantum measurements.
        \item $\widetilde{\cD}$ is invariant under swapping of Hilbert spaces. 
        \item $\widetilde{\cD}$ is non-increasing under the probabilistic mixing with a fixed set of quantum channels.\label{Coro:Properties_of_d_tilde}
    \end{enumerate}
\end{corollary}

Now, consider two sets of measurements $\cM=\{M_{i}\in\mathscr{M}(\cH)\}$ and $\cN=\{N_{i}\in\mathscr{M}(\cH)\}$ where $\Omega_{M_i}=\Omega_{N_i}=X_i\times Y_i~\forall i$ where $X_i$s and $Y_i$s are some sets of natural numbers. Then define the action of the map $\mathbbm{Merg}_{\cY}$ on $\cM$ as

\begin{align}
    \mathbbm{Merg}_{\cY}(\cM)=\{[\mathbbm{Merg}_{\cY}(\cM)]_i:=\mathbf{Merg}_{Y_i}(M_i)\},
\end{align}
where $\cY=\{Y_i\}$. Clearly, $\mathbbm{Merg}_{\cY}$ is the merginalization process for a set of measurements.

\begin{lemma}
    $\widetilde{\cD}$ is non-increasing under merginalization. In other words, for two sets of measurements $\cM$ and $\cN$
    \begin{align}
        \tilde{\cD}(\mathbbm{Merg}_{\cY}(\cM),\mathbbm{Merg}_{\cY}(\cN))\leq \tilde{\cD}(\cM,\cN).
    \end{align}
       \label{Le:dist_non_increase_merg}
\end{lemma}
\begin{proof}
    Let  
    \begin{align}
        &\widetilde{\cD}(\mathbbm{Merg}_{\cY}(\cM),\mathbbm{Merg}_{\cY}(\cN))\nonumber\\
        =&\cD_{\Diamond}(\Gamma_{\mathbf{Merg}_{Y_{i^*}}(M_{i^*})},\Gamma_{\mathbf{Merg}_{Y_{i^*}}(N_{i^*})}).
    \end{align} 
    Then
    \begin{align}
        \widetilde{\cD}(\mathbbm{Merg}_{\cY}(\cM),\mathbbm{Merg}_{\cY}(\cN))=&\cD_{\Diamond}(\Gamma_{\mathbf{Merg}_{Y_{i^*}}(M_{i^*})},\Gamma_{\mathbf{Merg}_{Y_{i^*}}(N_{i^*})})\nonumber\\
        \leq&\cD_{\Diamond}(\Gamma_{M_{i^*}},\Gamma_{N_{i^*}})\nonumber\\
        \leq&\max_i\cD_{\Diamond}(\Gamma_{M_i},\Gamma_{N_i})\nonumber\\
        =&\widetilde{\cD}(\cM,\cN),
    \end{align}
    where in the second line, we have used Eq. \eqref{Eq:diam_non_inc_merg}.
\end{proof}

At this point, we would like to mention the following observations

\begin{itemize}
    \item Consider a set of measurements $\cM=\{M_i\in\mathscr{M}(\cH)\}$. An arbitrary free transformation $\cW_{incomp}$ of the resource theory of measurement incompatibility acting on $\cM$  can be written as
    \begin{align}
        \cW_{incomp}[\cM]_j(y)=&\sum_{\mu}p(\mu)\sum_{i,x,l}r(y|i,x,\mu,j,l) \nonumber\\
        &q(i|\mu,j,l)\tilde{\Phi}^{\dagger}_{l|\mu}[M_i(x)],\label{Eq:meas_incomp_trans}
    \end{align}
    where $\cI_{\mu}=\{\tilde{\Phi}_{l|\mu}:\cL(\cK)\rightarrow\cL(\cH)\}$ are the quantum instruments for all $\mu$, with $p(\mu|j)$, $r(y|i,x,\mu,j,l)$, and $q(i|\mu,j,l)$ being the conditional probabilities \cite{Buscemi_meas_incomp}. Clearly, $\cW_{incomp}[\cM]_j$ is the $j$-th measurement of the transformed set $\cW_{incomp}[\cM]\subset \mathscr{M}(\cK)$. Now, consider the quantum instrument $\cJ=\{\Phi_{\lambda}=p(\mu)\tilde{\Phi}^{\dagger}_{l|\mu}\}$ where $\lambda=(\mu,l)$ Then the eq. \eqref{Eq:meas_incomp_trans} can be simplified as

    \begin{align}
        \cW_{incomp}[\cM]_j(y)=&\sum_{\lambda,i,x}r(y|i,x,\lambda,j) \nonumber\\
        &q(i|\lambda,j)\Phi^{\dagger}_{\lambda}[M_i(x)].\label{Eq:meas_incomp_trans_simp}
    \end{align}
    
    Now, one can show that
    \begin{align}
        \Gamma_{\cW_{incomp}[\cM]_j}=\Theta_j\circ(\Sigma_{\cM}\otimes\mathbbm{I}_{\cH_{\Lambda}})\circ\Psi_j\circ\cE,\label{Eq:meas_prep_channel_incomp_trans}
    \end{align}
    where $\cE(\rho)\in\mathscr{C}(\cK,\cH\otimes\cH_{\Lambda})$ such that $\cE(\rho)=\sum_{\lambda}\Phi_{\lambda}(\rho)\otimes\ket{\lambda}\bra{\lambda}$ for all $\rho\in\cL(\cK)$ where $\{\ket{\lambda}\}$ forms an orthonormal basis of Hilbert space $\cH_{\Lambda}$, for all $j$, $\Psi_j\in\mathscr{C}(\cH\otimes\cH_{\Lambda},\cH\otimes\cH_I\otimes\cH_{\Lambda})$ such that $\Psi_j(\sigma)=\sum_{i,\lambda}q(i\mid \lambda,j)\ket{i,\lambda}\bra{\lambda}\sigma\ket{\lambda}\bra{i,\lambda}$ for all $\rho\in\cL(\cH\otimes\cH_{\Lambda})$ where $\{\ket{i}\}$ forms an orthonormal basis of Hilbert space $\cH_{I}$, and for all $j$, $\Theta_j\in\mathscr{C}(\cH_A\otimes\cH_I\otimes\cH_{\Lambda},\cH_{B})$ such that $\Theta_j(\omega)=\sum_{y,x,i,\lambda}r(y|i,x,\lambda,j)\ket{y}\bra{x,i,\lambda}\omega\ket{x,i,\lambda}\bra{y}$ for all $\omega\in\cL(\cH_A\otimes\cH_I\otimes\cH_{\Lambda})$ where $\{\ket{y}\}$ forms an orthonormal basis of Hilbert space $\cH_{B}$.

    Clearly, Eq. \eqref{Eq:meas_prep_channel_incomp_trans} is a special case of Eq. \eqref{eq:far_gen_tra}. Therefore, from Eq. \eqref{Eq:def_dist_meas_of_set_meas} and Theorem \ref{theor:mono_set_tran}, we obtain

    \begin{align}
        \widetilde{\cD}(\cW_{incomp}[M_1],\cW_{incomp}[M_2])
        \leq&\widetilde{\cD}(M_1,M_2).
    \end{align}

    \item Consider a measurement $M=\{M(x)\}\in\mathscr{M}(\cH)$. An arbitrary free transformation $\cW_{sharp}$ of the resource theory of measurement sharpness acting on $M$ can be written as
    \begin{align}
        \cW_{sharp}[M](x)=\mu\Lambda^{\dagger}(M(x))+(1-\mu)p(x)\Id_{\cK},
    \end{align}
    where $p=\{p(x)\}$ is a probability distribution \cite{Buscemi_meas_sharp}. 

    Now, let $N:=\{ N(x)=p(x)\Id_{\cK}\}$. Then 
    \begin{align}
        \Gamma_{\cW_{sharp}[M]}=\cP_{\mu,N}[\Gamma_{M}\circ\Lambda].
    \end{align}
    As pre-processing $\Lambda$ is a special case of eq. \eqref{eq:far_gen_tra}, for any two measurements $M_1$ and $M_2$, from Eq. \eqref{Eq:def_dist_meas_of_set_meas}, Theorem \ref{theor:mono_set_tran} and Corollary \ref{Coro:mono_set_prob_mix} we obtain 
    \begin{align}
        \widetilde{\cD}(\cW_{sharp}[M_1],\cW_{sharp}[M_2])=&\overline{\cD}(\{\Gamma_{\cW_{sharp}[M_1]}\},\{\Gamma_{\cW_{sharp}[M_2]}\})\nonumber\\
        \leq&\overline{\cD}(\{\Gamma_{M_1}\},\{\Gamma_{M_2}\})\nonumber\\
        =&\widetilde{\cD}(M_1,M_2).
    \end{align}

    \item An arbitrary free transformation $\cW_{coh}$ of the resource theory of measurement coherence acting on $M$ can be written as

    \begin{align}
        \cW_{coh}[M](x)=&\Lambda_{coh}^{\dagger}(M(x))\nonumber\\
        =&\sum_{j}K^{\dagger}_jM(x)K_j,
    \end{align}
    where $K_j=\sum_k c_{k,i}\ket{\pi_k(i)}\bra{i}$, $\pi_k$ is the permutation, $\{\ket{i}\}$ is the orthonormal basis, and $\sum_jK^{\dagger}_jK_j=\Id_{\cH}$  \cite{baek_meas_coh}. Clearly, $\Gamma_{\cW_{coh}[M]}=\Gamma_{M}\circ\Lambda_{coh}$. The pre-processing $\Lambda_{coh}$ is a special case of eq. \eqref{eq:far_gen_tra}. Therefore, from Theorem \ref{theor:mono_set_tran}, and Eq. \eqref{Eq:def_dist_meas_of_set_meas}, we obtain 

    \begin{align}
        \widetilde{\cD}(\cW_{sharp}[M_1],\cW_{coh}[M_2])
        \leq&\widetilde{\cD}(M_1,M_2).
    \end{align}
\end{itemize}

Therefore, $\widetilde{\cD}$ is monotononically non-increasing under free transformation of (1) the resource theory of measurement incompatibility, (2) the resource theory of measurement sharpness, and (3) the resource theory of measurement coherence.

From next the subsection, we assume a generic distance $\mathbbm{D}$ that satisfies following properties

\begin{enumerate}
    \item[D1.] $\mathbbm{D}$ satisfies joint convexity property.
    \item[D2.] $\mathbbm{D}$ is jointly subadditive under tensor product and is invariant under the tensor product with same set of quantum measurements.
    \item[D3.] $\mathbbm{D}$ is non-increasing (at least) under the free transformation of a given resource theory.
    \item[D4.] $\mathbbm{D}$ is invariant under swapping of Hilbert spaces.
    \item[D5.] $\mathbbm{D}$ is non-increasing under merginalization.
\end{enumerate}
From this section, we already know that there exist the distance $\widetilde{\cD}$ that satisfies Properties [D1], [D2] and [D4] (from Corollary \ref{Coro:Properties_of_d_tilde}), Property [D5] (from Lemma \ref{Le:dist_non_increase_merg}) for an arbitrary resource theory of measurements, [D3] (from above-said observations) at least for the resource theory of measurement incompatibility, measurement sharpness, and measurement coherence . From next subsection, the use of a generic distance $\mathbbm{D}$ will help us to keep our analysis more general.

\subsection{Distance-based measures for measurement-based resources}
\label{Subsec:dist_based_res_meas}
Now, as an example, we show that from the distance measure $\mathbbm{D}$, one can create a resource measure for a given resource theory of measurements. We also derive some properties of such a measure.

Consider an arbitrary set of measurements $\cM=\{M_i\in\mathscr{M}(\cH_A)\}$. Then the enlarged version of $\cM$ is defined as
$\hat{\cM}_{\cH_A\otimes\cH_B}=\{(\hat{M}_i)_{\cH_A\otimes\cH_B}\}$. Clearly, $\hat{\cM}_{\cH_A\otimes\cH_B}=\cM\otimes\cT_{\cH_B}$. Here, $\cT_{\cH_B}$ is considered as a set with one element.

Now, we define the resource measure

\begin{align}
    \overline{\mathbbm{R}}(\cM)=\inf_{\cH_B}\min_{\cN_{\cH_A\otimes\cH_B}\in\cF_{\cH_A\otimes\cH_B}}\mathbbm{D}(\hat{\cM}_{\cH_A\otimes\cH_B},\cN_{\cH_A\otimes\cH_B}), \label{Eq:def_dist_res_meas}
\end{align}
where $\cF_{\cH_A\otimes\cH_B}$ is the set of free sets of measurements (obviously with the same number of elements as the set $\hat{\cM}_{\cH_A\otimes\cH_B}$) acting on the Hilbert space $\cH_A\otimes\cH_B$ and $\cH_B$ being an arbitrary finite dimensional Hilbert space. We have used the suffix $\cH_A\otimes\cH_B$ to denote the Hilbert space the measurements act on. Throughout this subsection, we will stick to this particular notation.

\begin{theorem}
    $\overline{\mathbbm{R}}$ is a resource measure for a generic measurement-based resource theory.
\end{theorem}

\begin{proof}

    Let, for an arbitrary $\cH_B$, 

    \begin{align}
        \mathbbm{K}(\cM,\cH_B):=&\min_{\cN_{\cH_A\otimes\cH_B}\in\cF_{\cH_A\otimes\cH_B}}\mathbbm{D}(\hat{\cM}_{\cH_A\otimes\cH_B},\cN_{\cH_A\otimes\cH_B}) \\
        =&\mathbbm{D}(\hat{\cM}_{\cH_A\otimes\cH_B},\cN^*_{\cH_A\otimes\cH_B}),
    \end{align}
    where $\cN^*_{\cH_A\otimes\cH_B}\in\cF_{\cH_A\otimes\cH_B}$ is a free set of measurements. Note that $\overline{\mathbbm{R}}(\cM)\geq 0$ for all $\cM$ as $\mathbbm{K}(\cM,\cH_B)\geq 0$ for all $\cM$ and for all $\cH_B$. Now, if $\cM\notin\cF_{\cH_A}$, then by assumption [A1], we have $\hat{\cM}_{\cH_A\otimes\cH_B}\notin\cF_{\cH_A\otimes\cH_B}~\forall \cH_B$. Then for all $\cH_B$, we have $\mathbbm{K}(\cM,\cH_B)>0$ and therefore, from Eq. \ref{Eq:def_dist_res_meas}, we have $\overline{\mathbbm{R}}>0$. Now, assume, $\cM\in\cF_{\cH_A}$. Then as $\mathbbm{D}(\cM,\cM)=0$ from Eq. \ref{Eq:def_dist_res_meas}, we have $\overline{\mathbbm{R}}=0$.
    
    Consider a free transformation $\cW$ that transforms the set of measurements $\cM\subset\mathscr{M}(\cH_A)$ to another set of measurements $\cW[M]\subset\mathscr{M}(\cK_A)$. Now, for an arbitrary $\cH_B$
    \begin{align}
        \mathbbm{K}(\cM,\cH_B)=&\mathbbm{D}(\hat{\cM}_{\cH_A\otimes\cH_B},\cN^*_{\cH_A\otimes\cH_B})\\
        \geq& \mathbbm{D}(\cW\otimes\cI^{\dagger}_{\cH_B}[\hat{\cM}_{\cH_A\otimes\cH_B}],\cW\otimes\cI^{\dagger}_{\cH_B}[\cN^*_{\cH_A\otimes\cH_B}])\nonumber\\
        \geq& \min_{\cN_{\cH_B}\in\cF_{\cK_A\otimes\cH_B}}\mathbbm{D}(\hat{\cW[\cM]}_{\cK_A\otimes\cH_B},\cN_{\cK_A\otimes\cH_B})\nonumber\\
        =& \mathbbm{K}(\cW[\cM],\cH_B),
    \end{align}
    where in second line, we have used the fact that $\cW\otimes\cI^{\dagger}_{\cH_B}$ is a free transformation as $\cW$ is a free transformation as given by assumptions [A2] and [A3], and in the third line, we used the obvious fact that $\cW\otimes\cI^{\dagger}_{\cH_B}[\hat{\cM}_{\cH_A\otimes\cH_B}]=\hat{\cW[\cM]}_{\cK_A\otimes\cH_B}$ as $\hat{\cM}_{\cH_A\otimes\cH_B}=\cM\otimes\cT_{\cH_B}$.
    Therefore,

    \begin{align}
        \mathbbm{K}(\cM,\cH_B)\geq& \mathbbm{K}(\cW[\cM],\cH_B)~\forall \cH_B\nonumber\\
        or,~\inf_{\cH_B}\mathbbm{K}(\cM,\cH_B)\geq& \inf_{\cH_B}\mathbbm{K}(\cW[\cM],\cH_B)\nonumber\\
        or,~\overline{\mathbbm{R}}(\cM)\geq& \overline{\mathbbm{R}}(\cW[\cM]).
    \end{align}
\end{proof}

\begin{remark}
    \rm{Note that usually a distance-based resource measure for a set of measurement $\cM$ is defined as $\overline{\mathbf{R}}(\cM)=\min_{\cN\in\cF_{\cH}}\mathbbm{D}(\cM,\cN)$ \cite{Tendick_dist_res_meas}. But in Eq. \eqref{Eq:def_dist_res_meas}, we have defined the resource measure in a slightly more general way and our measure is consistent with many different aspects (including tensor product structure) of measurement-based resource theories that we show in the rest of this subsection.}
\end{remark}

\begin{proposition}
    $\overline{\mathbbm{R}}(\cM\otimes\cT_{\cH_B})=\overline{\mathbbm{R}}(\cM)$ for any $\cM=\{M_i\in\mathscr{M}(\cH_A)\}$. \label{Proposi:dist_meas_inv_add_triv_meas}
\end{proposition}

\begin{proof}
For the simplicity of notation in the following arguments let us denote $\cM\otimes\cT_{\cH_B}$ as $\cM^{\prime} $. Then, we have $\hat{\cM^{\prime}}_{\cH_A\otimes\cH_B\otimes\cH_C}=\hat{\cM}_{\cH_A\otimes\cH_B\otimes\cH_C}$ and therefore we can write the following,

\begin{align}
&\overline{\mathbbm{R}}(\cM^{\prime})\nonumber\\
=&\inf_{\cH_C}\min_{\cN_{\cH_A\otimes\cH_B\otimes\cH_C}\in\cF_{\cH_A\otimes\cH_B\otimes\cH_C}}\mathbbm{D}(\hat{\cM^{\prime}}_{\cH_A\otimes\cH_B\otimes\cH_C},\cN_{\cH_A\otimes\cH_B\otimes\cH_C})\nonumber\\
\geq&\inf_{\cH_B\otimes\cH_C}\min_{\cN_{\cH_A\otimes(\cH_B\otimes\cH_C)}\in\cF_{\cH_A\otimes(\cH_B\otimes\cH_C)}}\mathbbm{D}(\hat{\cM}_{\cH_A\otimes\cH_B\otimes\cH_C},\cN_{\cH_A\otimes\cH_B\otimes\cH_C})\nonumber\\
=&\overline{\mathbbm{R}}(\cM).
\end{align}

Now, let  

\begin{align}
    \mathbbm{K}(\cM,\cH_C)=&\min_{\cN_{\cH_A\otimes\cH_C}\in\cF_{\cH_A\otimes\cH_C}}\mathbbm{D}(\hat{\cM}_{\cH_A\otimes\cH_C},\cN_{\cH_A\otimes\cH_C})\nonumber\\
    =&\mathbbm{D}(\hat{\cM}_{\cH_A\otimes\cH_C},\cN^*_{\cH_A\otimes\cH_C}),
\end{align}

Let us denote the elements of $\cN^*_{\cH_A\otimes\cH_C}$ as $N^*_j$ i.e., $\cN^*_{\cH_A\otimes\cH_C}=\{N^*_j\in\mathscr{M}(\cH_A\otimes\cH_C)\}$. \\

Then for an arbitrary $\cH_C$ and arbitrary $\cH_B$,
\begin{align}
    \mathbbm{K}(\cM,\cH_C)=&\mathbbm{D}(\hat{\cM}_{\cH_A\otimes\cH_C},\cN^*_{\cH_A\otimes\cH_C})\nonumber\\
    =&\mathbbm{D}(\hat{\cM}_{\cH_A\otimes\cH_C}\otimes\cT_{\cH_B},\cN^*_{\cH_A\otimes\cH_C}\otimes\cT_{\cH_B})\nonumber\\
    =&\mathbbm{D}(\hat{\cM}_{\cH_A\otimes\cH_C\otimes\cH_B},\cN^*_{\cH_A\otimes\cH_C}\otimes\cT_{\cH_B})\nonumber\\
   =&\mathbbm{D}(\hat{\cM}_{\cH_A\otimes\cH_B\otimes\cH_C},\cN^{\prime}_{\cH_A\otimes\cH_B\otimes\cH_C})\nonumber\\
   \geq&\min_{\cN_{\cK_{ABC}}\in\cF_{\cK_{ABC}}}\mathbbm{D}(\hat{\cM}_{\cK_{ABC}},\cN_{\cK_{ABC}})\nonumber\\
   =&\mathbbm{K}(\cM\otimes\cT_{\cH_B},\cH_C),
\end{align}
where in fourth line,
\begin{equation}
    \cN^{\prime}_{\cH_A\otimes\cH_B\otimes\cH_C}:= \{N^{\prime}_j =\mathbbm{I}^{\dagger}_{\cH_A}\otimes\mathbf{SWAP}^{\dagger}_{\cH_B\leftrightarrow \cH_C}[N^*_j\otimes\cT_{\cH_B}]\},\nonumber
\end{equation}
 clearly, 
$\mathbbm{I}^{\dagger}_{\cH_A}\otimes\mathbf{SWAP}^{\dagger}_{\cH_B\leftrightarrow \cH_C}[\hat{\cM}_{\cH_A\otimes\cH_C\otimes\cH_B}]=\hat{\cM}_{\cH_A\otimes\cH_B\otimes\cH_C}$, and we have used Property [D4], and in the fifth line $\cK_{ABC}=\cH_A\otimes\cH_B\otimes\cH_C$, and we have used the assumptions [A1], [A4], and [A5] (and therefore, $\cN^{\prime}_{\cH_A\otimes\cH_B\otimes\cH_C}$ is a free object).

Hence,

\begin{align}
    \mathbbm{K}(\cM,\cH_C)\geq&\mathbbm{K}(\cM\otimes\cT_{\cH_B},\cH_C)~\forall \cH_C\nonumber\\
    or,~\inf_{\cH_C}\mathbbm{K}(\cM,\cH_C)\geq&\inf_{\cH_C}\mathbbm{K}(\cM\otimes\cT_{\cH_B},\cH_C)\nonumber\\
    or,~\overline{\mathbbm{R}}(\cM)\geq&\overline{\mathbbm{R}}(\cM\otimes\cT_{\cH_B})
\end{align}
Therefore,
\begin{align}
\overline{\mathbbm{R}}(\cM)=\overline{\mathbbm{R}}(\cM\otimes\cT_{\cH_B}).
\end{align}
\end{proof}

\begin{proposition}
   $\overline{\mathbbm{R}}(\cM\otimes\cM^{\prime})\leq\overline{\mathbbm{R}}(\cM)+\overline{\mathbbm{R}}(\cM^{\prime})$ for any $\cM=\{M_i\in\mathscr{M}(\cH_A)\}$ and $\cM^{\prime}=\{M^{\prime}_i\in\mathscr{M}(\cH_B)\}$. In other words, $\overline{\mathbbm{R}}$ is subadditive under tensor product. \label{Proposi:res_mes_ten_sub_addit}
\end{proposition}

\begin{proof}
Let us consider,

\begin{align}
    \mathbbm{K}(\cM,\cH_C)=&\min_{\cN_{\cH_A\otimes\cH_C}\in\cF_{\cH_A\otimes\cH_C}}\mathbbm{D}(\hat{\cM}_{\cH_A\otimes\cH_C},\cN_{\cH_A\otimes\cH_C})\nonumber\\
    =&\mathbbm{D}(\hat{\cM}_{\cH_A\otimes\cH_C},\bar{\cN}_{\cH_A\otimes\cH_C}),
\end{align}

and 

\begin{align}
    \mathbbm{K}(\cM^{\prime},\cH_D)=&\min_{\cN^{\prime}_{\cH_B\otimes\cH_D}\in\cF_{\cH_B\otimes\cH_D}}\mathbbm{D}(\hat{\cM}^{\prime}_{\cH_B\otimes\cH_D},\cN_{\cH_B\otimes\cH_D})\nonumber\\
    =&\mathbbm{D}(\hat{\cM}^{\prime}_{\cH_B\otimes\cH_D},\bar{\cN}^{\prime}_{\cH_B\otimes\cH_D}).
\end{align}

 We denote the elements of  $\bar{\cN}_{\cH_A\otimes\cH_C}$ as $N_i$ i.e., $\bar{\cN}_{\cH_A\otimes\cH_C}=\{N_i\in\mathscr{M}(\cH_A\otimes\cH_C)\}$ and the elements of  $\bar{\cN}^{\prime}_{\cH_B\otimes\cH_D}$ as $N^{\prime}_i$ i.e., $\bar{\cN}^{\prime}_{\cH_B\otimes\cH_D}=\{N^{\prime}_j\in\mathscr{M}(\cH_B\otimes\cH_D)\}$. For notational simplicity let us now denote $\cM\otimes\cM^{\prime}$ by $\cM^{\prime\prime}$.\\
 
Then,
\begin{align}
&\mathbbm{K}(\cM,\cH_C)+\mathbbm{K}(\cM^{\prime},\cH_D)\nonumber\\
=&\mathbbm{D}(\hat{\cM}_{\cH_A\otimes\cH_C},\bar{\cN}_{\cH_A\otimes\cH_C})+\mathbbm{D}(\hat{\cM}^{\prime}_{\cH_B\otimes\cH_D},\bar{\cN}^{\prime}_{\cH_B\otimes\cH_D})\nonumber\\
\geq& \mathbbm{D}(\hat{\cM}_{\cH_A\otimes\cH_C}\otimes\hat{\cM}^{\prime}_{\cH_B\otimes\cH_D},\bar{\cN}_{\cH_A\otimes\cH_C}\otimes\bar{\cN}^{\prime}_{\cH_B\otimes\cH_D})\nonumber\\
=& \mathbbm{D}(\hat{\cM}^{\prime\prime}_{\cH_A\otimes\cH_B\otimes\cH_C\otimes\cH_D},\tilde{\cN}_{\cH_A\otimes\cH_B\otimes\cH_C\otimes\cH_D})\nonumber\\
\geq&\min_{\cN_{\cK_{ABCD}}\in\cF_{\cK_{ABCD}}}\mathbbm{D}(\hat{\cM}^{\prime\prime}_{\cK_{ABCD}},\cN_{\cK_{ABCD}})\nonumber\\
=&\mathbbm{K}(\cM^{\prime\prime},\cH_C\otimes\cH_D)\nonumber\\
=&\mathbbm{K}(\cM\otimes\cM^{\prime},\cH_C\otimes\cH_D),
\end{align}
where in the third line we have used the property [D2], in the fourth line we have made the substitution, 
\begin{equation}
    \tilde{\cN}_{\cH_A\otimes\cH_B\otimes\cH_C\otimes\cH_D}:=\{\tilde{N}_{ij}=\mathbbm{I}^{\dagger}_{\cH_A}\otimes\mathbf{SWAP}^{\dagger}_{\cH_B\leftrightarrow \cH_C}\otimes\mathbbm{I}^{\dagger}_{\cH_D}[N_i\otimes N^{\prime}_j]\},
\end{equation}
    
 as well as used property [D4], in the fifth line we defined $\cK_{ABCD}:=\cH_A\otimes\cH_B\otimes\cH_C\otimes\cH_D$ and also have used the assumptions [A1] and [A4] (and therefore, $\tilde{\cN}_{\cH_A\otimes\cH_B\otimes\cH_C\otimes\cH_D}$ is a free object).
Now, note that both $\mathbbm{K}(\cM,\cH_C)$ and $\mathbbm{K}(\cM^{\prime},\cH_D)$ are positive real numbers and therefore, 
$\overline{\mathbbm{R}}(\cM)+\overline{\mathbbm{R}}(\cM^{\prime})=\inf_{\cH_C}\mathbbm{K}(\cM,\cH_C)+\inf_{\cH_D}\mathbbm{K}(\cM^{\prime},\cH_D)=\inf_{\cH_C\otimes\cH_D}[\mathbbm{K}(\cM,\cH_C)+\mathbbm{K}(\cM^{\prime},\cH_D)]$. Then we can write,
\begin{align}
&\mathbbm{K}(\cM,\cH_C)+\mathbbm{K}(\cM^{\prime},\cH_D)\nonumber\\&\hspace{2cm}\geq\mathbbm{K}(\cM\otimes\cM^{\prime},\cH_C\otimes\cH_D)~\forall \cH_C,\cH_D\nonumber\\
or,&\inf_{\cH_C\otimes\cH_D}[\mathbbm{K}(\cM,\cH_C)+\mathbbm{K}(\cM^{\prime},\cH_D)]\nonumber\\&\hspace{2cm}\geq\inf_{\cH_C\otimes\cH_D}\mathbbm{K}(\cM\otimes\cM^{\prime},\cH_C\otimes\cH_D)\nonumber\\
or,&~\inf_{\cH_C}\mathbbm{K}(\cM,\cH_C)+\inf_{\cH_D}\mathbbm{K}(\cM^{\prime},\cH_D)\nonumber\\
&\hspace{2cm}\geq\inf_{\cH_C\otimes\cH_D}\mathbbm{K}(\cM\otimes\cM^{\prime},\cH_C\otimes\cH_D)\nonumber\\
or,&~\overline{\mathbbm{R}}(\cM)+\overline{\mathbbm{R}}(\cM^{\prime})\geq\overline{\mathbbm{R}}(\cM\otimes\cM^{\prime}).
\end{align}
\end{proof}

\begin{proposition}
For a convex resource theory,    $\overline{\mathbbm{R}}(p\cM+(1-p)\cM^{\prime})\leq p\overline{\mathbbm{R}}(\cM)+(1-p)\overline{\mathbbm{R}}(\cM^{\prime})$ for any $\cM=\{M_i\in\mathscr{M}(\cH_A)\}$ and $\cM^{\prime}=\{M^{\prime}_i\in\mathscr{M}(\cH_A)\}$ and any $0\leq p\leq 1$.
\end{proposition}

\begin{proof}
Let 

\begin{align}
    \mathbbm{K}(\cM,\cH_C)=&\min_{\cN_{\cH_A\otimes\cH_C}\in\cF_{\cH_A\otimes\cH_C}}\mathbbm{D}(\hat{\cM}_{\cH_A\otimes\cH_C},\cN_{\cH_A\otimes\cH_C})\nonumber\\
    =&\mathbbm{D}(\hat{\cM}_{\cH_A\otimes\cH_C},\bar{\cN}_{\cH_A\otimes\cH_C}),
\end{align}

and 

\begin{align}
    \mathbbm{K}(\cM^{\prime},\cH_C)=&\min_{\cN^{\prime}_{\cH_A\otimes\cH_C}\in\cF_{\cH_A\otimes\cH_C}}\mathbbm{D}(\hat{\cM}^{\prime}_{\cH_A\otimes\cH_C},\cN_{\cH_A\otimes\cH_C})\nonumber\\
    =&\mathbbm{D}(\hat{\cM}^{\prime}_{\cH_A\otimes\cH_C},\bar{\cN}^{\prime}_{\cH_A\otimes\cH_C}).
\end{align}

Again let $\cM^{\prime\prime}:=p\cM+(1-p)\cM^{\prime}$. Then for an arbitrary $\cH_C$

\begin{align}
    &p\mathbbm{K}(\cM,\cH_C)+(1-p)\mathbbm{K}(\cM^{\prime},\cH_C)\nonumber\\
    =&p\mathbbm{D}(\hat{\cM}_{\cH_A\otimes\cH_C},\bar{\cN}_{\cH_A\otimes\cH_C})+(1-p)\mathbbm{D}(\hat{\cM}^{\prime}_{\cH_A\otimes\cH_C},\bar{\cN}^{\prime}_{\cH_A\otimes\cH_C})\nonumber\\
    \geq& \mathbbm{D}(p\hat{\cM}_{\cH_A\otimes\cH_C}+(1-p)\hat{\cM}^{\prime}_{\cH_A\otimes\cH_C}, \nonumber \\ & \hspace{3.6cm} p\bar{\cN}_{\cH_A\otimes\cH_C}+(1-p)\bar{\cN}^{\prime}_{\cH_A\otimes\cH_C})\nonumber\\
    \geq& \min_{\cN_{\cH_A\otimes\cH_C}\in\cF_{\cH_A\otimes\cH_C}}\mathbbm{D}(\hat{\cM^{\prime\prime}}_{\cH_A\otimes\cH_C}, \cN_{\cH_A\otimes\cH_C})\nonumber\\
    =&\mathbbm{K}(\cM^{\prime\prime},\cH_C)\nonumber\\
    =&\mathbbm{K}(p\cM+(1-p)\cM^{\prime},\cH_C)
\end{align}
where in the third line, we have used the Property [D1], in the fourth line, we used the obvious fact that $\hat{\cM^{\prime\prime}}_{\cH_A\otimes\cH_C}=p\hat{\cM}_{\cH_A\otimes\cH_C}+(1-p)\hat{\cM}^{\prime}_{\cH_A\otimes\cH_C}$, and we have used the fact that as the resource theory is convex, free objects form a convex set and therefore, $p\bar{\cN}_{\cH_A\otimes\cH_C}+(1-p)\bar{\cN}^{\prime}_{\cH_A\otimes\cH_C}$ is free object.

Then
\begin{align}
    &p\mathbbm{K}(\cM,\cH_C)+(1-p)\mathbbm{K}(\cM^{\prime},\cH_C)\nonumber\\
    &\hspace{2cm}\geq\mathbbm{K}(p\cM+(1-p)\cM^{\prime},\cH_C)\nonumber\\
    or,~&\inf_{\cH_C}[p\mathbbm{K}(\cM,\cH_C)+(1-p)\mathbbm{K}(\cM^{\prime},\cH_C)]\nonumber\\
    &\hspace{2cm}\geq\inf_{\cH_C}\mathbbm{K}(p\cM+(1-p)\cM^{\prime},\cH_C)\nonumber\\
    or,~&\inf_{\cH_C}p\mathbbm{K}(\cM,\cH_C)+\inf_{\cH_C}(1-p)\mathbbm{K}(\cM^{\prime},\cH_C)\nonumber\\
    &\hspace{2cm}\geq\inf_{\cH_C}\mathbbm{K}(p\cM+(1-p)\cM^{\prime},\cH_C)\nonumber\\
    or,~&p\overline{\mathbbm{R}}(\mathcal{M})+(1-p)\overline{\mathbbm{R}}(\mathcal{M}^{\prime})\nonumber\\
    &\hspace{2cm}\geq\overline{\mathbbm{R}}(p\cM+(1-p)\cM^{\prime}).
\end{align}

Where in the second line we have used the fact that
$\overline{\mathbbm{R}}(\cM)+\overline{\mathbbm{R}}(\cM^{\prime})=\inf_{\cH_C}\mathbbm{K}(\cM,\cH_C)+\inf_{\cH_D}\mathbbm{K}(\cM^{\prime},\cH_D)=\inf_{\cH_C\otimes\cH_D}[\mathbbm{K}(\cM,\cH_C)+\mathbbm{K}(\cM^{\prime},\cH_D)]$ as both $\mathbbm{K}(\cM,\cH_C)$ and $\mathbbm{K}(\cM^{\prime},\cH_D)$ are positive real numbers.
\end{proof}

\begin{proposition} \label{prop:free}
    If marginalization is a free transformation (i.e., if assumption [A6] holds) for a given resource theory, then, $\overline{\mathbbm{R}}(\cM\otimes \cM^{\prime})\geq\overline{\mathbbm{R}}(\cM)$ for any $\cM=\{M_i\in\mathscr{M}(\cH_A)\}$ and any $\mathcal{N}=\{N_i\in\mathscr{M}(\mathcal{H}_B)\}$. \label{Proposi:dist_meas_merg_non_increas}
\end{proposition}

\begin{proof}
    For notational simplicity let us define $\cM^{\prime\prime}:=\cM\otimes\cM^{\prime}$ and $\cK_{ABC}:=\cH_A\otimes\cH_B\otimes\cH_C$.
    
   Again let 

   \begin{align}
       \mathbbm{K}(\cM\otimes\cM^{\prime},\cH_C)=&\min_{\cN_{\cK_{ABC}}\in\cF_{\cK_{ABC}}}\mathbbm{D}(\hat{\cM}^{\prime\prime}_{\cK_{ABC}},\cN_{\cK_{ABC}})\nonumber\\
       =&\mathbbm{D}(\hat{\cM}^{\prime\prime}_{\cK_{ABC}},\bar{\cN}_{\cK_{ABC}}).
   \end{align}

   Then for an arbitrary $\cH_C$

    \begin{align}
        &\mathbbm{K}(\cM\otimes\cM^{\prime},\cH_C)\nonumber\\
        =&\mathbbm{D}(\hat{\cM}^{\prime\prime}_{\cK_{ABC}},\bar{\cN}_{\cK_{ABC}})\nonumber\\
        \geq&\mathbbm{D}(\mathbbm{Merg}_{\Omega_{\cM^{\prime}}}(\hat{\cM}^{\prime\prime}_{\cK_{ABC}}),\mathbbm{Merg}_{\Omega_{\cM^{\prime}}}(\bar{\cN}_{\cK_{ABC}}))\nonumber\\
        =&\mathbbm{D}(\hat{\cM}_{\cK_{ABC}},\mathbbm{Merg}_{\Omega_{\cM^{\prime}}}(\bar{\cN}_{\cK_{ABC}}))\nonumber\\
        \geq&\min_{\cN_{\cK_{ABC}}\in\cF_{\cK_{ABC}}}\mathbbm{D}(\hat{\cM}_{\cK_{ABC}},\cN_{\cK_{ABC}})\nonumber\\
        =&\mathbbm{K}(\cM,\cH_B\otimes\cH_C).
    \end{align}
    where in the third line, we have used the Property [D5] and in the fourth line we have used the obvious fact that $\mathbbm{Merg}_{\Omega_{\cM^{\prime}}}(\hat{\cM}^{\prime\prime}_{\cK_{ABC}})=\hat{\cM}_{\cK_{ABC}}$, and in the fifth line, we have used the fact that the marginalization is a free transformation (i.e., if assumption [A6] holds) and therefore, $\mathbbm{Merg}_{\Omega_{\cM^{\prime}}}(\bar{\cN}_{\cK_{ABC}})$ is a free object.

    Therefore,
    \begin{align}
        &\mathbbm{K}(\cM\otimes\cM^{\prime},\cH_C)\geq\mathbbm{K}(\cM,\cH_B\otimes\cH_C)~\forall \cH_C\nonumber\\
         or,~&\inf_{\cH_C}\mathbbm{K}(\cM\otimes\cM^{\prime},\cH_C)\geq\inf_{\cH_C}\mathbbm{K}(\cM,\cH_B\otimes\cH_C)\nonumber\\
         or,~&\overline{\mathbbm{R}}(\cM\otimes\cM^{\prime})\geq\inf_{\cH_C}\mathbbm{K}(\cM,\cH_B\otimes\cH_C)\nonumber\\
         or,~&\overline{\mathbbm{R}}(\cM\otimes\cM^{\prime})\geq\inf_{\cH_B\otimes\cH_C}\mathbbm{K}(\cM,\cH_B\otimes\cH_C)\nonumber\\
         or,~&\overline{\mathbbm{R}}(\cM\otimes\cM^{\prime})\geq\overline{\mathbbm{R}}(\cM) \label{Eq:dist_meas_ten_merg_dec}.
    \end{align}
    
\end{proof}

\begin{corollary}\label{cor:free}
    If marginalization is a free transformation (i.e., if assumption [A6] holds) for a given resource theory, then, we have $\overline{\mathbbm{R}}(\cM\otimes \cN)=\overline{\mathbbm{R}}(\cM)$ for any $\cM=\{M_i\in\mathscr{M}(\cH_A)\}$ and any $\cN=\{N_i\in\mathscr{M}(\cH_B)\}\in\cF_{\cH_B}$.
\end{corollary}

\begin{proof}
  As $\cN\in\cF_{\cH_B}$, $\overline{\mathbbm{R}}(\cN)=0$. Therefore, from Proposition \ref{Proposi:res_mes_ten_sub_addit}, we get,
  \begin{align}
      \overline{\mathbbm{R}}(\cM\otimes\cN)\leq\overline{\mathbbm{R}}(\cM).
  \end{align}

Again from Proposition \ref{Proposi:dist_meas_merg_non_increas}, we have
 \begin{align}
      \overline{\mathbbm{R}}(\cM\otimes\cN)\geq\overline{\mathbbm{R}}(\cM).
  \end{align}
  
  Therefore, 
   \begin{align}
      \overline{\mathbbm{R}}(\cM\otimes\cN)=\overline{\mathbbm{R}}(\cM)\label{Eq:dist_meas_add_free_inv}.
  \end{align}
\end{proof}

Note that the Eq. \eqref{Eq:dist_meas_ten_merg_dec} and Eq. \eqref{Eq:dist_meas_add_free_inv} may hold even if the assumption [A6] does not hold as the statements of Proposition \ref{prop:free} and Corollary \ref{cor:free} are only sufficiency conditions.



\begin{proposition}
    If $\mathbbm{D}=\widetilde{\cD}$, then for any set of measurements $\cM\subset \mathscr{M}(\cH_A)$, we have
       \begin{align}
    \overline{\mathbbm{R}}(\cM)\leq&\min\left\{\frac{2\mathscr{R}(\cM)}{1+\mathscr{R}(\cM)},\frac{2\mathscr{W}(\cM)}{1+\mathscr{W}(\cM)}\right\}.    
    \end{align}
 \end{proposition}

\begin{proof}
We have $\mathbbm{D}=\widetilde{\cD}$. Suppose in Eq. \eqref{Eq:def_robust}, the minimum occurs for $\tilde{\cM}=\cM^*=\{M_i^*\}$ and $ \mathscr{R}(\cM)=r^*$. Then  
\begin{align}
&N_i=\frac{M_i}{1+r^*}+ \frac{r^*M^*_i}{1+r^*}~\forall i\nonumber\\
&or,~N_i\otimes\cT_{\cB}=\frac{M_i\otimes\cT_{\cB}}{1+r^*}+ \frac{r^*\tilde{M}_i\otimes\cT_{\cB}}{1+r^*}~\forall i,\cH_B\nonumber\\
&or,~(\hat{N}_i)_{\cH_A\otimes\cH_B}=\frac{(\hat{M}_i)_{\cH_A\otimes\cH_B}}{1+r^*}+ \frac{r^*(\hat{M}^*_i)_{\cH_A\otimes\cH_B}}{1+r^*}~\forall i,\cH_B\nonumber\\
&or,~(\hat{M}_i)_{\cH_A\otimes\cH_B}-(\hat{N}_i)_{\cH_A\otimes\cH_B}\nonumber\\
&\hspace{2cm}=\frac{r^*[(\hat{M}_i)_{\cH_A\otimes\cH_B}-(\hat{M}^*_i)_{\cH_A\otimes\cH_B}]}{1+r^*}~\forall i,\cH_B\nonumber\\
&or,~\cD_{\Diamond}((\hat{M}_i)_{\cH_A\otimes\cH_B},(\hat{N}_i)_{\cH_A\otimes\cH_B})\nonumber\\
&\hspace{2cm}=\frac{r^*}{1+r^*}\cD_{\Diamond}((\hat{M}_i)_{\cH_A\otimes\cH_B},(\hat{M}^*_i)_{\cH_A\otimes\cH_B})~\forall i,\cH_B\nonumber\\
&or,~\cD_{\Diamond}((\hat{M}_i)_{\cH_A\otimes\cH_B},(\hat{N}_i)_{\cH_A\otimes\cH_B})\nonumber\\
&\hspace{2cm}\leq\frac{2r^*}{1+r^*}~\forall i,\cH_B\label{Eq:diam_robust_ineq}
\end{align}
In the last inequality, we have used the fact that $\cD_{\Diamond}((\hat{M}_i)_{\cH_A\otimes\cH_B},(\hat{M}^*_i)_{\cH_A\otimes\cH_B})\leq 2$ (see Remark \ref{Re:diam_dist_half_fact}). Let 
\begin{align}
\widetilde{\cD}(\hat{M}_{\cH_A\otimes\cH_B},\hat{\cN}_{\cH_A\otimes\cH_B})=\cD_{\Diamond}((\hat{M}_{i^*})_{\cH_A\otimes\cH_B},(\hat{N}_{i^*})_{\cH_A\otimes\cH_B})\label{Eq:ov_dist_robust},
\end{align}
and 
\begin{align}
    \mathbbm{K}(\cM,\cH_B)=&\min_{\cN_{\cH_A\otimes\cH_B}\in\cF_{\cH_A\otimes\cH_B}}\widetilde{\cD}(\hat{\cM}_{\cH_A\otimes\cH_C},\cN_{\cH_A\otimes\cH_B})\label{Eq:K_robust}.
\end{align}

Then observing (from  Eq.\eqref{Eq:def_robust}, the assumption [A1] and the assumption [A5]) that $(\hat{N}_{i^*})_{\cH_A\otimes\cH_B}\in\cF_{\cH_A\otimes\cH_B}$, from Eq. \eqref{Eq:diam_robust_ineq}, Eq. \eqref{Eq:ov_dist_robust} and Eq.  \eqref{Eq:K_robust}, we have for all $\cH_B$
\begin{align}
    \mathbbm{K}(\cM,\cH_B)&\leq \widetilde{\cD}(\hat{M}_{\cH_A\otimes\cH_B},\hat{\cN}_{\cH_A\otimes\cH_B})\nonumber\\
    &= \cD_{\Diamond}((\hat{M}_{i^*})_{\cH_A\otimes\cH_B},(\hat{N}_{i^*})_{\cH_A\otimes\cH_B})\nonumber\\
    &\leq\frac{2r^*}{1+r^*}.
\end{align}
Then recalling $\mathscr{R}(\cM)=r^*$, noting the fact that $\mathbbm{D}=\overline{\cD}$ and from Eq. \eqref{Eq:def_dist_res_meas}, we have
\begin{align}
\label{loosebound}
    \overline{\mathbbm{R}}(\cM)&\leq \mathbbm{K}(\cM,\cH_B)\nonumber\\
    &\leq\frac{2\mathscr{R}(\cM)}{1+\mathscr{R}(\cM)}.
\end{align}
Suppose in Eq. \eqref{Eq:def_weight}, the minimum occurs for $\tilde{\cM}=\cM^*=\{M_i^*\}$, $\tilde{\cN}=\cN^*=\{N_i^*\}$ and $ \mathscr{W}(\cM)=r^*$. Then  
\begin{align}
&M_i=\frac{N^*_i}{1+r}+ \frac{r^*M^*_i}{1+r}~\forall i\nonumber\\
&or,~M_i\otimes\cT_{\cB}=\frac{N^*_i\otimes\cT_{\cB}}{1+r}+ \frac{r^*M^*_i\otimes\cT_{\cB}}{1+r}~\forall i,\cH_B\nonumber\\
&or,~(\hat{M}_i)_{\cH_A\otimes\cH_B}
=\frac{(\hat{N}^*_i)_{\cH_A\otimes\cH_B}}{1+r}+ \frac{r^*(\hat{M}^*_i)_{\cH_A\otimes\cH_B}}{1+r}~\forall i,\cH_B\nonumber\\
&or,~(\hat{M}_i)_{\cH_A\otimes\cH_B}-(\hat{N}_i)_{\cH_A\otimes\cH_B}\nonumber\\
&\hspace{2cm}=\frac{r^*[(\hat{M}^*_i)_{\cH_A\otimes\cH_B}-(\hat{N}^*_i)_{\cH_A\otimes\cH_B}]}{1+r^*}~\forall i,\cH_B\nonumber\\
&or,~\cD_{\Diamond}((\hat{M}_i)_{\cH_A\otimes\cH_B},(\hat{N}_i)_{\cH_A\otimes\cH_B})\nonumber\\
&\hspace{2cm}=\frac{r^*}{1+r^*}\cD_{\Diamond}((\hat{M}^*_i)_{\cH_A\otimes\cH_B},(\hat{N}^*_i)_{\cH_A\otimes\cH_B})~\forall i,\cH_B\nonumber\\
&or,~\cD_{\Diamond}((\hat{M}_i)_{\cH_A\otimes\cH_B},(\hat{N}_i)_{\cH_A\otimes\cH_B})\nonumber\\
&\hspace{2cm}\leq\frac{2r^*}{1+r^*}~\forall i,\cH_B\label{Eq:diam_weight_ineq}
\end{align}
In the last inequality, we have used the fact that $\cD_{\Diamond}((\hat{M}^*_i)_{\cH_A\otimes\cH_B},(\hat{N}^*_i)_{\cH_A\otimes\cH_B})\leq 2$. Let 
\begin{align}
\widetilde{\cD}(\hat{M}_{\cH_A\otimes\cH_B},\hat{\cN}_{\cH_A\otimes\cH_B})=\cD_{\Diamond}((\hat{M}_{i^*})_{\cH_A\otimes\cH_B},(\hat{N}_{i^*})_{\cH_A\otimes\cH_B})\label{Eq:ov_dist_weight},
\end{align}
and 
\begin{align}
    \mathbbm{K}(\cM,\cH_B)=&\min_{\cN_{\cH_A\otimes\cH_B}\in\cF_{\cH_A\otimes\cH_B}}\widetilde{\cD}(\hat{\cM}_{\cH_A\otimes\cH_C},\cN_{\cH_A\otimes\cH_B})\label{Eq:K_weight}.
\end{align}
Then (from Eq. \eqref{Eq:def_weight}, the assumption [A1] and the assumption [A5]) that $(\hat{N}_{i^*})_{\cH_A\otimes\cH_B}\in\cF_{\cH_A\otimes\cH_B}$, from Eq. \eqref{Eq:diam_weight_ineq}, Eq. \eqref{Eq:ov_dist_weight} and Eq.  \eqref{Eq:K_weight}, we have for all $\cH_B$
\begin{align}
    \mathbbm{K}(\cM,\cH_B)&\leq \widetilde{\cD}(\hat{M}_{\cH_A\otimes\cH_B},\hat{\cN}_{\cH_A\otimes\cH_B})\nonumber\\
    &= \cD_{\Diamond}((\hat{M}_{i^*})_{\cH_A\otimes\cH_B},(\hat{N}_{i^*})_{\cH_A\otimes\cH_B})\nonumber\\
    & \leq \frac{2r^*}{1+r^*}.
\end{align}
Then recalling $\mathscr{W}(\cM)=r^*$, noting the fact that $\mathbbm{D}=\widetilde{\cD}$ and from Eq. \eqref{Eq:def_dist_res_meas}, we have
\begin{align}
\label{looseboundW}
    \overline{\mathbbm{R}}(\cM)&\leq \mathbbm{K}(\cM,\cH_B)\nonumber\\
    &\leq\frac{2\mathscr{W}(\cM)}{1+\mathscr{W}(\cM)}.
\end{align}
Therefore, we finally have 
\begin{align}
\label{eq:tightboundRW}
\overline{\mathbbm{R}}(\cM)\leq&\min\left\{\frac{2\mathscr{R}(\cM)}{1+\mathscr{R}(\cM)},\frac{2\mathscr{W}(\cM)}{1+\mathscr{W}(\cM)}\right\}.
\end{align}
\end{proof}
Clearly, Eq.~\eqref{eq:tightboundRW} constitute a tighter bound of $\overline{\mathbbm{R}}(\cM)$ than those given in Eq.~\eqref{loosebound} and Eq.~\eqref{looseboundW}. Now, it is well known that resource robustness and resource weight are directly linked to operational tasks such as state discrimination \cite{Skrzypczyk_incomp_state_disc} and exclusion \cite{Uola_q_res_exclus}. Therefore, through Proposition 6, we try to \emph{relate} our distance-based resource measure to these operational tasks.

\subsection{Some properties of $\epsilon$-measures of measurement-based resources}
\label{Subsec:properties_epsilon_measure}
In this section, we study the mathematical properties for $\epsilon$-measures for generic measurement-based resource measure $\mathbbm{R}$ satisfying properties [R1]-[R3]. We start with the following important theorem.

\begin{theorem}
    If $\mathbbm{R}$ is a measurement-based resource monotone then $\mathbbm{R}^{\mathbbm{D}}_{inf,\epsilon}$ is also a resource monotone for any distance measure $\mathbbm{D}$.
\end{theorem}

\begin{proof}
    If $\mathbbm{R}$ is a resource monotone then $\mathbbm{R}(\cW[\cM])\leq \mathbbm{R}(\cM)$ for any free transformation $\cW$ and any set of measurements $\cM$. Let $\cN^*$ be the set of measurements for which the minimum occurs in Eq. \eqref{Eq:def_eps_inf}. Therefore, $\mathbbm{D}(\cM,\cN^*)\leq\epsilon$ and so that we have, $\mathbbm{D}(\cW(\cM),\cW(\cN^*))\leq\mathbbm{D}(\cM,\cN^*)\leq\epsilon$ (by the assumption [D3]). We can then write, 
    \begin{align}
        \mathbbm{R}^{\mathbbm{D}}_{inf,\epsilon}(\cM)&=\mathbbm{R}(\cN^*)\nonumber\\
        &\geq \mathbbm{R}(\cW(\cN^*))\nonumber\\
        &\geq \mathbbm{R}^{\mathbbm{D}}_{inf,\epsilon}(\cW(\cM)).
    \end{align}
    
\end{proof}

\begin{proposition}
    For any convex resource measure $\mathbbm{R}$, the $\epsilon$-measure $\mathbbm{R}^{\mathbbm{D}}_{inf,\epsilon}$ is also a convex resource measure for any distance $\mathbbm{D}$ that is jointly convex.
\end{proposition}

\begin{proof}

Let $\cN_{1}^{*}$ and $\cN_{2}^{*}$ are the sets of measurements for which we get minimums in Eq.\eqref{Eq:def_eps_inf} corresponding to $\mathbbm{R}^{\mathbbm{D}}_{inf,\epsilon}(\cM_{1})$ and $\mathbbm{R}^{\mathbbm{D}}_{inf,\epsilon}(\cM_{2})$ respectively. Clearly, we then have $\mathbbm{D}(\cM_{1},\cN_{1}^{*})\leq\epsilon$ and $\mathbbm{D}(\cM_{2},\cN_{2}^{*})\leq\epsilon$. Now, from the joint convexity of the distance $\mathbbm{D}$ i.e., from the assumption [D1], we obtain

\begin{align}
    &\mathbbm{D}(p\cM_{1}+(1-p)\cM_{2},p\cN_{1}^{*}+(1-p\cN_{2}^{*}))\nonumber\\
    \leq& p\mathbbm{D}(\cM_{1},\cN_{1}^{*})+(1-p)\mathbbm{D}(\cM_{2},\cN_{2}^{*})\nonumber\\
    \leq& p\epsilon+(1-p)\epsilon\nonumber\\
    =&\epsilon\label{Eq:gen_dist_conv_eps}.
\end{align}

Then,
    \begin{align}
        &p\mathbbm{R}^{\mathbbm{D}}_{inf,\epsilon}(\cM_{1})+(1-p)\mathbbm{R}^{\mathbbm{D}}_{inf,\epsilon}(\cM_{2}) \nonumber\\
        &= p\mathbbm{R}(\cN_{1}^{*}))+(1-p)\mathbbm{R}(\cN_{2}^{*})) \nonumber\\
          &\geq \mathbbm{R}( p(\cN_{1}^{*})+(1-p)\cN_{2}^{*}) \nonumber\\
        &\geq \mathbbm{R}^{\mathbbm{D}}_{inf,\epsilon}( p(\cM_1)+(1-p)\cM_2),
    \end{align}
\end{proof}
where the first inequality is due to the convexity of resource measure $\mathbbm{R}$, whereas the second inequality follows from Eq. \eqref{Eq:gen_dist_conv_eps} and Eq. \eqref{Eq:def_eps_inf}.

  \begin{proposition}\label{prop:subadepsilon}
    For any resource measure $\mathbbm{R}$ that is sub-additive under tensor product, and for any $\epsilon=\epsilon_{1}+\epsilon_{2}$ the following relation is satisfied:
    \begin{equation}
       \mathbbm{R}^{\mathbbm{D}}_{inf,\epsilon_{1}}(\cM_{1})+\mathbbm{R}^{\mathbbm{D}}_{inf,\epsilon_{2}}(\cM_{2}) \geq \mathbbm{R}^{\mathbbm{D}}_{inf,\epsilon}(\cM_{1}\otimes\cM_{2}),
    \end{equation} 
    where $\mathbbm{D}$ is a generic distance satisfying joint subadditivity under tensor product.\\
\end{proposition}

\begin{proof}

Let $\cN_{1}^{*}$ and $\cN_{2}^{*}$ be the sets of measurements for which minimums of Eq.\eqref{Eq:def_eps_inf} occur corresponding to $\mathbbm{R}^{\mathbbm{D}}_{inf,\epsilon_{1}}(\cM_{1})$ and $\mathbbm{R}^{\mathbbm{D}}_{inf,\epsilon_{2}}(\cM_{2})$ respectively.

Now, in order to prove the proposition, let us first note that for any generic distance measure $\mathbbm{D}$ that satisfies joint subadditivity under tensor product the following relation holds.

\begin{align}\label{eq:subDepsilon}
    \mathbbm{D}(\cM_{1}\otimes\cM_{2}, \cN_{1}^{*}\otimes\cN_{2}^{*})&\leq\mathbbm{D}(\cM_{1}, \cN_{1}^{*})+\mathbbm{D}(\cM_{2}, \cN_{2}^{*}) \nonumber \\
    &\leq \epsilon_{1}+\epsilon_{2}=\epsilon,
\end{align}

where the last inequality is due to the fact that $\cN_{1}^{*}$ and $\cN_{2}^{*}$ are those sets of measurements for which minimums of Eq.\eqref{Eq:def_eps_inf} occur corresponding to $\mathbbm{R}^{\mathbbm{D}}_{inf,\epsilon_{1}}(\cM_{1})$ and $\mathbbm{R}^{\mathbbm{D}}_{inf,\epsilon_{2}}(\cM_{2})$ respectively.

Now we have,
\begin{align}
  \mathbbm{R}^{\mathbbm{D}}_{inf,\epsilon_{1}}(\cM_{1})+\mathbbm{R}^{\mathbbm{D}}_{inf,\epsilon_{2}}(\cM_{2}) &= \mathbbm{R}(\cN_{1}^{*})+\mathbbm{R}(\cN_{2}^{*}) \nonumber\\
  &\geq\mathbbm{R}(\cN_{1}^{*}\otimes\cN_{2}^{*}),
\end{align}
where the inequality is due to the subadditivity of the resource measure $\mathbbm{R}$ under tensor product. Now, due to Eq. \eqref{eq:subDepsilon} we can write, 

\begin{align}
    \mathbbm{R}(\cN_{1}^{*}\otimes\cN_{2}^{*})&\geq \inf_{\substack{\mathbbm{D}(\cM_{1}\otimes\cM_{2}, \cN_{1}\otimes\cN_{2})\leq \epsilon\\ ~|\cM_{i}|=|\cN_{i}|, ~i\in \{1,2\}} } \mathbbm{R}(\cN_{1}\otimes\cN_{2})\nonumber\\
    &= \mathbbm{R}^{\mathbbm{D}}_{inf,\epsilon}(\cM_{1}\otimes\cM_{2}).
\end{align}


Combination of the above two equations proves the proposition.

\end{proof}

\begin{proposition}
    Consider a generic distance $\mathbbm{D}$ and a resource measure $\mathbbm{R}$ such that two sets of measurements $\cM_{1}$ and  $\cM_{2}$ have a distance $\mathbbm{D}(\cM_{1},\cM_{2})=t$. Then for $\mathbbm{R}^{\mathbbm{D}}_{inf,\epsilon_{i}}(\cM_{i})=\mathbbm{R}(\cN_{i}^{*})$ and a convex mixture $\cS_{p}=(1-p)\cM_{1}+p\cN_{2}^{*}$ of the sets of measurements  with $0 \leq p \leq \frac{\epsilon_{1}}{t+\epsilon_{2}}$ the following relation holds:
    \begin{equation}
       \mathbbm{R}^{\mathbbm{D}}_{inf,\epsilon_{1}}(\cM_{1})-\mathbbm{R}^{\mathbbm{D}}_{inf,\epsilon_{2}}(\cM_{2}) \leq (1-p)[\mathbbm{R}(\cM_{2})-\mathbbm{R}^{\mathbbm{D}}_{inf,\epsilon_{2}}(\cM_{2})].
    \end{equation} 

\end{proposition}

\begin{proof}
    From the joint convexity property of the distance $\mathbbm{D}$ we can write,

    \begin{align}
        \mathbbm{D}(\cM_{1},\cS_{p})
        &= \mathbbm{D}(\cM_{1},(1-p)\cM_{1}+p\cN_{2}^{*})\nonumber\\
        &\leq p \mathbbm{D}(\cM_{1},\cN_{2}^{*})\nonumber \\
        &\leq p ( \mathbbm{D}(\cM_{1},\cM_{2})+  \mathbbm{D}(\cM_{2},\cN_{2}^{*})).
    \end{align}
The first inequality is due to the joint convexity of $\mathbbm{D}$. Now, as $\mathbbm{D}(\cM_{1},\cM_{2})=t$ and $\mathbbm{D}(\cM_{2},\cN_{2}^{*}) \leq \epsilon_{2}$ we can write $\mathbbm{D}(\cM_{1},\cM_{2})+ \mathbbm{D}(\cM_{2},\cN_{2}^{*})\leq (t+\epsilon_{2})$. Furthermore, for $p$ in the range $0 \leq p \leq \frac{\epsilon_{1}}{t+\epsilon_{2}}$ we have $p (t+\epsilon_{2}) \leq \epsilon_{1} $. We then finally have $ \mathbbm{D}(\cM_{1},\cS_{p})\leq \epsilon_{1}$, which further implies,

\begin{align}
    \mathbbm{R}^{\mathbbm{D}}_{inf,\epsilon_{1}}(\cM_{1}) &\leq \mathbbm{R}(\cS_{p}) \nonumber \\
    &\leq (1-p) \mathbbm{R}(\cM_{1}) + p \mathbbm{R}(\cN_{2}^{*})\nonumber\\
    &= (1-p) \mathbbm{R}(\cM_{1}) + p \mathbbm{R}^{\mathbbm{D}}_{inf,\epsilon_{1}}(\cN_{2}^{*}).
\end{align}
The second line comes due to the convexity of $\mathbbm{R}$ (i.e., assumption [R3]). From the above relation we can finally write,
 \begin{equation}\label{eq:cont1}
       \mathbbm{R}^{\mathbbm{D}}_{inf,\epsilon_{1}}(\cM_{1})-\mathbbm{R}^{\mathbbm{D}}_{inf,\epsilon_{2}}(\cM_{2}) \leq (1-p)[\mathbbm{R}(\cM_{1})-\mathbbm{R}^{\mathbbm{D}}_{inf,\epsilon_{2}}(\cM_{2})].
    \end{equation} 
This completes the proof.
\end{proof}

\begin{corollary}
The $\epsilon$-measure $\mathbbm{R}^{\mathbbm{D}}_{inf,\epsilon_{1}}$ is i) continuous for any original convex resource measure $\mathbbm{R} $, and, ii) continuous functional of $\epsilon$.
\end{corollary}

\begin{proof}
 i) The continuity of the $\mathbbm{R}^{\mathbbm{D}}_{inf,\epsilon_{1}}$ demands that if two sets of measurements $\cM_{1}$ and $\cM_{2}$ are very close, i.e., $\mathbbm{D}(\cM_{1},\cM_{2})=t \rightarrow 0$ then the corresponding $\epsilon$-measures are also very close, i.e., $\lvert\mathbbm{R}^{\mathbbm{D}}_{inf,\epsilon}(\cM_{1})-\mathbbm{R}^{\mathbbm{D}}_{inf,\epsilon}(\cM_{2})\rvert\rightarrow 0$. In order to prove this, we choose that, for $\epsilon_{1}=\epsilon_{2}=\epsilon$ and $p=\frac{\epsilon}{\epsilon+t}$. Then for $\cM_{1}\neq \cM_{2}$, from Eq. \eqref{eq:cont1} we can write,

\begin{equation}\label{eq:cont2}
       \mathbbm{R}^{\mathbbm{D}}_{inf,\epsilon}(\cM_{1})-\mathbbm{R}^{\mathbbm{D}}_{inf,\epsilon}(\cM_{2}) \leq \frac{t}{\epsilon+t}[\mathbbm{R}(\cM_{1})-\mathbbm{R}^{\mathbbm{D}}_{inf,\epsilon}(\cM_{2})].
    \end{equation} 

 The above equation also implies,

\begin{equation}\label{eq:cont3}
       \mathbbm{R}^{\mathbbm{D}}_{inf,\epsilon}(\cM_{1})-\mathbbm{R}^{\mathbbm{D}}_{inf,\epsilon}(\cM_{2}) \geq -\frac{t}{\epsilon+t}[\mathbbm{R}(\cM_{2})-\mathbbm{R}^{\mathbbm{D}}_{inf,\epsilon}(\cM_{1})].
    \end{equation}

 Therefore, the following holds,

 \begin{align}\label{eq:cont4}
       &\lvert\mathbbm{R}^{\mathbbm{D}}_{inf,\epsilon}(\cM_{1})-\mathbbm{R}^{\mathbbm{D}}_{inf,\epsilon}(\cM_{2})\rvert\nonumber \\
       &\leq \frac{t}{\epsilon+t}\max\Big\{(\mathbbm{R}(\cM_{2})-\mathbbm{R}^{\mathbbm{D}}_{inf,\epsilon}(\cM_{1})), \nonumber \\ 
       & \hspace{3.5cm}(\mathbbm{R}(\cM_{1})-\mathbbm{R}^{\mathbbm{D}}_{inf,\epsilon}(\cM_{2}))\Big\}.
    \end{align} 

This finally implies that $\lvert\mathbbm{R}^{\mathbbm{D}}_{inf,\epsilon}(\cM_{1})-\mathbbm{R}^{\mathbbm{D}}_{inf,\epsilon}(\cM_{2})\rvert\rightarrow 0$ as $t \rightarrow 0$.

ii) In order to prove that the $\epsilon$-measure $\mathbbm{R}^{\mathbbm{D}}_{inf,\epsilon_{1}}$ is continuous functional of $\epsilon$, we choose $\cM_{1}=\cM_{2}=\cM$ (i.e., $t=0$) and $p=\frac{\epsilon_{1}}{\epsilon_{2}}$.  Then for $\epsilon_{2}\geq \epsilon_{1}\geq 0 $ from Eq. \eqref{eq:cont1} we can write,

\begin{equation}
       \mathbbm{R}^{\mathbbm{D}}_{inf,\epsilon_{1}}(\cM)-\mathbbm{R}^{\mathbbm{D}}_{inf,\epsilon_{2}}(\cM) \leq \frac{\epsilon_{2}-\epsilon_{1}}{\epsilon_{2}}[\mathbbm{R}(\cM)-\mathbbm{R}^{\mathbbm{D}}_{inf,\epsilon_{2}}(\cM)].
    \end{equation} 

The above equation also implies,

\begin{equation}
       \mathbbm{R}^{\mathbbm{D}}_{inf,\epsilon_{1}}(\cM)-\mathbbm{R}^{\mathbbm{D}}_{inf,\epsilon_{2}}(\cM) \geq -\frac{\epsilon_{2}-\epsilon_{1}}{\epsilon_{2}}[\mathbbm{R}(\cM)-\mathbbm{R}^{\mathbbm{D}}_{inf,\epsilon_{2}}(\cM)].
    \end{equation} 
Therefore, we can immediately write,

\begin{align}
       &\lvert\mathbbm{R}^{\mathbbm{D}}_{inf,\epsilon_{1}}(\cM)-\mathbbm{R}^{\mathbbm{D}}_{inf,\epsilon_{2}}(\cM) \rvert \nonumber \\ &\leq \frac{\epsilon_{2}-\epsilon_{1}}{\epsilon_{2}} \max\{[\mathbbm{R}(\cM)-\mathbbm{R}^{\mathbbm{D}}_{inf,\epsilon_{2}}(\cM)], \nonumber \\
       & \hspace{2.5cm}[\mathbbm{R}^{\mathbbm{D}}_{inf,\epsilon_{2}}(\cM)-\mathbbm{R}(\cM)] \}.
    \end{align} 

It is then easy to see from the above equation that if $\epsilon_{2}-\epsilon_{1}\rightarrow 0$, then $\mathbbm{R}^{\mathbbm{D}}_{inf,\epsilon_{1}}(\cM)-\mathbbm{R}^{\mathbbm{D}}_{inf,\epsilon_{2}}(\cM) \rightarrow 0$. This concludes the proof.

\end{proof}

\subsection{One-shot dilution cost, smooth asymptotic resource measures and its connection to measurement-based $\epsilon$- measures}
\label{Subsec:one_shot_distill_assymp_res_meas}
Two important subclasses of resource manipulation task are resource distillation and dilution \cite{RegulaDis,winterdis16,Liudis,LiuDetDis19}. Resource dilution is the process of transforming an element from a chosen set of reference resource objects into the desired resource object via free transformations while utilizing minimum possible reference resources. For instance, in a measurement-based resource theory, resource dilution involves transforming a target resource $\cN$-a set of measurements, into a set $\cM$ within the reference set of measurements $\cZ$, where $\cO$ denotes the set of all free transformations. Clearly, $\cM$ is an element of $\cZ$. The performance of this task is quantified the one-shot dilution cost and defined \cite{PRXQuantumRegula} as,
\begin{align}\label{eq:dilutionc}
   \mathbbm{C}^{\cZ}_{\epsilon}(\cM) & = \inf_{\substack{\mathbbm{D}(\cM,\cW(\cN)) \leq \epsilon\\ \cN \in \cZ, \cW \in \cO,~|\cM|=|\cN|} } \mathbbm{R}(\cN) 
\end{align}
\begin{proposition}
    The one shot resource delusion cost in a measurement based resource theory is lower bounded by the epsilon measures corresponding to the measurement-based concerned resource, i.e., $\mathbbm{C}^{\cZ}_{\epsilon}(\cM) \geq  \mathbbm{R}^{\mathbbm{D}}_{inf,\epsilon}(\cM)$ for all $\cM\in\mathscr{M}(\cH)$.
\end{proposition}
\begin{proof}
    To prove the proposition, observe that from Eq.~\eqref{eq:dilutionc} we have 
$\mathbbm{D}(\cM,\cW(\cN)) \leq \epsilon$, which implies that 
$\cW(\cN)\subset \mathscr{M}(\cH)$. Furthermore, since $\cW$ is a free transformation, it follows that 
$\mathbbm{R}(\cN)\geq \mathbbm{R}(\cW(\cN))$. Hence,
    \begin{align}
    \mathbbm{C}^{\cZ}_{\epsilon}(\cM)=& \inf_{\substack{\mathbbm{D}(\cM,\cW(\cN)) \leq \epsilon\\ \cN \in \cZ, \cW \in \cO,~|\cM|=|\cN|} } \mathbbm{R}(\cN)\nonumber\\
    \geq& \inf_{\substack{\mathbbm{D}(\cM,\cW(\cN)) \leq \epsilon\\ \cN \in \cZ, \cW \in \cO,~|\cM|=|\cN|} } \mathbbm{R}(\cW(\cN))\nonumber\\
    \geq& \inf_{\substack{\mathbbm{D}(\cM,\cN^{\prime}) \leq \epsilon\\ \cN^{\prime}\subset\mathscr{M}(\cH),~|\cM|=|\cN^{\prime}|} } \mathbbm{R}(\cN^{\prime})\nonumber\\
    =&\mathbbm{R}^{\mathbbm{D}}_{inf,\epsilon}(\cM).
    \end{align}
\end{proof}

It is then clear that $\epsilon$-measure of the measurement-based resources constitutes the lower bound for preparing a given set of measurement upto an error $\epsilon$ by free transformation. Therefore, for our case, this $\epsilon$-measure provides a guaranteed lower bound on how much resource is needed to approximately create a set of measurements showing its importance as a operational tool in one-shot resource dilution tasks.

Resource distillation is the process of converting multiple resource objects into a reference resource object, using free transformations, while aiming to maximize the resource content of the reference object. For instance, resource distillation consists of transforming the set $\cM$ of measurements into a desired target set of measurements $\cN$ via a free transformation $\cW \in \cO$. The figure of merit of this task is given by the distillable resource \cite{PRXQuantumRegula},
\begin{align}
   \label{eq:distill} 
   \mathbbm{E}^{\cZ}_{\epsilon}(\cM) &= \sup_{\substack{\mathbbm{D}(\cW(\cM),\cN) \leq \epsilon\\  \cN \in \cZ,\cW \in \cO,~|\cM|=|\cN|} } \mathbbm{R}(\cN), 
\end{align}
Now, analogous to the one-shot dilution cost we deduce a relationship between the distillable resource and the $\epsilon$-measure which expressed through the following proposition.


 
\begin{proposition}
    The distillable resource in a measurement based resource theory is lower bounded by the $\epsilon$-measures corresponding to the measurement-based concerned resource, i.e., $\mathbbm{E}^{\cZ}_{\epsilon}(\cM) \geq  \mathbbm{R}^{\mathbbm{D}}_{inf,\epsilon}(\cM)$ for all $\cM\in\mathscr{M}(\cH)$, with $\cZ=\mathscr{M}(\cH)$.
\end{proposition}

\begin{proof}
    As $\cW\in \cO$ is a free transformation, $\mathbbm{D}(\cM,\cN) \leq \mathbbm{D}(\cW(\cM),\cN) \leq \epsilon$ we have $\cW(\cN), \cW(\cM)\subset \mathscr{M}(\cH)$. Then,
    \begin{align}
    \mathbbm{E}^{\cZ}_{\epsilon}(\cM)=& \sup_{\substack{\mathbbm{D}(\cW(\cM),\cN) \leq \epsilon\\  \cN \in \cZ, \cW \in \cO,~|\cM|=|\cN|} } \mathbbm{R}(\cN)\nonumber\\
    \geq& \sup_{\substack{\mathbbm{D}(\cM,\cN) \leq \epsilon\\ \cN \subset \mathscr{M}(\cH), ~|\cM|=|\cN|} } \mathbbm{R}(\cN)\nonumber\\
    \geq& \inf_{\substack{\mathbbm{D}(\cM,\cN) \leq \epsilon\\ \cN \subset \mathscr{M}(\cH), ~|\cM|=|\cN|} } \mathbbm{R}(\cN)\nonumber\\
    =&\mathbbm{R}^{\mathbbm{D}}_{inf,\epsilon}(\cM).
    \end{align}
The second line is due to the fact that the set $\mathscr{S}:=\{N \in \cZ: \mathbbm{D}(\cW(\cM),\cN) \leq \epsilon, \cW \in \cO,~|\cM|=|\cN| \}$ contains the set $\mathscr{W}:=\{N \subset \mathscr{M}(\cH): \mathbbm{D}(\cM,\cN) \leq \epsilon, \cW \in \cO,~|\cM|=|\cN| \}$ due to the assumption [A3] as identity channel is a free transformation. Therefore, the former set is a superset of the latter, and hence the supremum of $\mathbbm{R}(\mathcal{N})$ is larger for the former.
\end{proof}

The one-shot dilution cost is only relevant when considering single set of measurements. However, for more general scenarios one quantifies how much resource is needed per copy (or can be extracted) when manipulating many copies of a set of measurements. Asymptotic resource measures quantify the rate at which a measurement based quantum resource can be consumed to prepare many copies of a target set of measurements, or can be extracted from many copies of a given resourceful set of measurements. We define the lower and upper regularization of resource measures for the $\epsilon$- measure of measurement-based resources $\mathbbm{R}$ as smooth regularization defined as,

\begin{align}
    \mathbbm{C}^{\infty}_{l}(\cM)&= \lim_{\epsilon \rightarrow 0^+}\liminf_{n\rightarrow \infty}  \frac{1}{n}\mathbbm{R}^{\mathbbm{D}}_{inf,\epsilon}(\cM^{\otimes n }), \label{Eq:smooth_low}\\
     \mathbbm{C}^{\infty}_{u}(\cM)&= \lim_{\epsilon \rightarrow 0^+}\limsup_{n\rightarrow \infty}\frac{1}{n}\mathbbm{R}^{\mathbbm{D}}_{inf,\epsilon}(\cM^{\otimes n }).\label{Eq:smooth_high}
\end{align}

We now describe one important result regarding the above smooth regularization of resource measures as follows:

\begin{proposition}
    The smooth regularization of resource measures are also resource measures. In other words

    \begin{align}
        \mathbbm{C}^{\infty}_{l}(\cM)&\geq  \mathbbm{C}^{\infty}_{l}(\cW(\cM)), \nonumber \\
        \mathbbm{C}^{\infty}_{u}(\cM)&\geq  \mathbbm{C}^{\infty}_{u}(\cW(\cM)).
    \end{align}
    and $\mathbbm{C}^{\infty}_{l}(\cM)=0$ and $\mathbbm{C}^{\infty}_{u}(\cM)=0$ if and only if $\cM$ is free.
\end{proposition}

\begin{proof}
   From Eq. \eqref{Eq:smooth_low}, we have,

   \begin{align}
       \mathbbm{C}^{\infty}_{l}(\cM)&= \lim_{\epsilon \rightarrow 0^+}\liminf_{n\rightarrow \infty}\frac{1}{n}\mathbbm{R}^{\mathbbm{D}}_{inf,\epsilon}(\cM^{\otimes n }) \nonumber \\
       &\geq \lim_{\epsilon \rightarrow 0^+}\liminf_{n\rightarrow \infty}\frac{1}{n}\mathbbm{R}^{\mathbbm{D}}_{inf,\epsilon}( \cW^{\otimes n}(\cM^{\otimes n })) \nonumber \\
       &= \lim_{\epsilon \rightarrow 0^+}\liminf_{n\rightarrow \infty}\frac{1}{n}\mathbbm{R}^{\mathbbm{D}}_{inf,\epsilon}( \cW(\cM)^{\otimes n }) \nonumber \\
       &=   \mathbbm{C}^{\infty}_{l}(\cW(\cM)),
   \end{align}
where in the second line we have used the assumption [A2].
   Therefore, it is clear that $ \mathbbm{C}^{\infty}_{l}$  is non-increasing under free transformations.

   Now, from assumption [A1], we know that $\cM$ is if and only if $\cM^{\otimes n }$ is free for an arbitrary $n$. Therefore, for an arbitrary $n$, $\mathbbm{R}^{\mathbbm{D}}_{inf,\epsilon}(\cM^{\otimes n })=0$ if and only if $\cM$ is free. Therefore, from Eq. \eqref{Eq:smooth_low}, we have  $\mathbbm{C}^{\infty}_{l}(\cM)=0$ if and only if $A$ is free. Hence, it constitutes a valid resource measure. Similarly, one can prove the same for upper regularization $ \mathbbm{C}^{\infty}_{u}$.
\end{proof} 
    

\vspace{1cm}

\section{Conclusion}
\label{Sec:Conc}
In this work, we have studied distance-based resource measures and 
$\epsilon-$measures within the framework of measurement-based resource theories. Specifically, we analyze the transformation of sets of measurements and quantum channels, define a suitable distance measure for these objects, and examine its mathematical properties. Building on this, we construct a distance-based resource measure based on a generic distance function satisfying certain natural properties. We further investigate key properties of the  $\epsilon-$measures for measurement-based resources. In addition, we explore the one-shot dilution cost and smooth asymptotic resource measures, highlighting their connections to $\epsilon-$measures in the context of measurement-based resources.

It is important to note that our results are valid not only for resource theories where the resource concerns a single measurement (e.g., measurement coherence, and measurement sharpness), but also for the resource theories where the resource involves a set of measurements (e.g., measurement incompatibility). Furthermore, we tried to keep our analysis as general as possible. More specifically, in Sec.\ref{Subsec:dist_based_res_meas}, our analysis is based on a generic distance function $\mathbbm{D}$ assumed to satisfy a set of natural properties, and in Sec.\ref{Subsec:properties_epsilon_measure}, we consider $\mathbbm{R}$ to be a generic resource measure for studying the properties of the $\epsilon-$measures.

As previously noted, measurement-based resource theories have received significantly less attention compared to their state-based counterparts, and to the best of our knowledge, the study of $\epsilon-$measures within this framework remains unexplored. Consequently, our work offers a broad scope and opens several promising avenues for future research. We outline a few of them here. First, an immediate direction is to identify examples of distance functions, beyond $\overline{\cD}$ that satisfies the properties [D1]-[D5]. Second, one can investigate the design of information-theoretic tasks whose performance can be directly quantified by resource measures derived from $\overline{\cD}$. Third, it is crucial to explore how, in scenarios where a set of measurements is only partially known up to some accuracy $\epsilon$, the $\epsilon-$measures can be used to meaningfully quantify the performance in relevant information-theoretic tasks. Fourth, There exist nonconvex resources, namely, certain other layers of nonclassicality \cite{heinosari16,mitra21} relevant to sets of measurements beyond incompatibility—whose resource quantification would be worthwhile to explore in the future. Furthermore, since state-based coherence is known to serve as a resource in quantum heat engines \cite{kwon23}, it will be interesting to investigate the potential applications of measurement-based coherence in quantum thermodynamics. This motivates further investigation.

\section{Acknowledgements}
 The authors acknowledge financial support from Institute of Information and communications Technology Planning and Evaluation (IITP) Grant (RS-2023-00222863, RS-2025-02292999). The authors also acknowledge Kyunghyun Baek and Nirman Ganguly for their valuable comments.

\bibliography{reference}

\end{document}